\documentclass[11pt]{article}
\usepackage{amsmath}
\usepackage{graphicx}
\usepackage{enumerate}
\usepackage{natbib}
\usepackage{url} 

\usepackage[normalem]{ulem}

\usepackage{color}

\usepackage{xcolor}

\usepackage{floatrow}
\usepackage{graphicx}
\graphicspath{{images/}}
\usepackage{caption}
\usepackage{subcaption}

\usepackage{dirtytalk}

\usepackage{newtxtext}
\usepackage[subscriptcorrection]{newtxmath}

\usepackage{amsthm}

\newtheorem{theorem}{Theorem}
\newtheorem{proposition}{Proposition}

\newtheorem{lemma}{Lemma}

\newtheorem{assumption}{Assumption}

\usepackage[plain,noend]{algorithm2e}

\makeatletter
\renewcommand{\algocf@captiontext}[2]{#1\algocf@typo. \AlCapFnt{}#2} 
\def\@algocf@capt@plain{top}
\renewcommand{\algocf@makecaption}[2]{%
  \addtolength{\hsize}{\algomargin}%
  \sbox\@tempboxa{\algocf@captiontext{#1}{#2}}%
  \ifdim\wd\@tempboxa >\hsize
    \hskip .5\algomargin%
    \parbox[t]{\hsize}{\algocf@captiontext{#1}{#2}}
  \else%
    \global\@minipagefalse%
    \hbox to\hsize{\box\@tempboxa}
  \fi%
  \addtolength{\hsize}{-\algomargin}%
}
\makeatother

\usepackage{dirtytalk}
\usepackage{comment}

\usepackage{hyperref}
\usepackage{xr-hyper}

\usepackage{multirow}

\newcommand{\blind}{1}

\usepackage{makecell}

\begin{document}

\def\spacingset#1{\renewcommand{\baselinestretch}%
{#1}\small\normalsize} \spacingset{1}


\if1\blind
{
  \title{\bf Inference on covariance structure in high-dimensional multi-view data}
  \author{Lorenzo Mauri
  \hspace{.2cm}\\
    Department of Statistical Science, Duke University\\
    \small{\texttt{lorenzo.mauri@duke.edu}}\\
    David B. Dunson\\
        Department of Statistical Science, Duke University\\
    }
  \maketitle
} \fi

\if0\blind
{
  \bigskip
  \bigskip
  \bigskip
  \begin{center}
    {\LARGE\bf Inference on covariance structure in high-dimensional multi-view data}
\end{center}
  \medskip
} \fi

\bigskip
\begin{abstract}
This article focuses on covariance estimation for multi-view data. Popular approaches rely on factor-analytic decompositions that have shared and view-specific latent factors. Posterior computation is conducted via expensive and brittle Markov chain Monte Carlo (MCMC) sampling or variational approximations that underestimate uncertainty and lack theoretical guarantees. Our proposed methodology employs spectral decompositions to estimate and align latent factors that are active in at least one view. Conditionally on these factors, we choose jointly conjugate prior distributions for factor loadings and residual variances. The resulting posterior is a simple product of normal-inverse gamma distributions for each variable, bypassing MCMC and facilitating posterior computation. We prove favorable increasing-dimension asymptotic properties, including posterior contraction and central limit theorems for point estimators. We show excellent performance in simulations, including accurate uncertainty quantification, and apply the methodology to integrate four high-dimensional views from a multi-omics dataset of cancer cell samples. 
\end{abstract}

\noindent%
\vfill

\newpage

\section{Introduction}
Multi-view data are routinely collected in many applied settings, ranging from biology to ecology. In biomedical and precision medicine applications, researchers commonly collect multi-omics data, ranging from gene expression to metabolomics. In such settings, there is often interest in inferring the within- and between-view covariance structure in the data. However, it is common for different data views to have very different dimensionalities, distributions, and signal-to-noise ratios, so that methods that ignore the multi-view structure and concatenate the data prior to analysis can produce sub-optimal results. 
This has motivated a rich and rapidly expanding literature on methods for integrating data from multiple sources, including \citet{jive, ajive,  mofa, mofa+, sjive, mefisto, coop_mw, yi_et_al, clic, jafar, stack_svd}. Factor models \citep{harman1976modern, West2003} provide a powerful way to reduce dimensionality and accurately estimate covariance in high dimensions. For multi-view data, a popular approach models the $i$-th observation from the $m$-th view as 
\begin{equation}\label{eq:mofa}
\begin{aligned}
      y_{mi} &=  \Lambda_m \eta_{mi} + \epsilon_{mi}, \quad \epsilon_{mi} \sim N_p(0, \Psi_m),  \\
     \eta_{mi}  &= A_m^\top \eta_i , \quad \eta_i \sim  N_{k_0}(0, I_{k_0}),
\end{aligned}
\quad  (i=1, \dots, n; m=1, \dots, M),
\end{equation}
where $ \Psi_m = \text{diag}(\sigma_{m1}^2, \dots, \sigma_{m p_m}^2)$ are residual variances. $A_m \in \mathbb  R^{k_0 \times k_m}$ is a matrix that selects the $k_m$ factors $\eta_{mi}$ that are active in the $m$-th view from the vector of all factors $\eta_i$. In particular, each element of $A_m$ is either $0$ or $1$, with each column having a unique non-zero entry and each row having at most one non-zero entry.

Model \eqref{eq:mofa} induces the following factorizations of the intra- and inter-view covariances
\begin{equation}\label{eq:covariances}
    cov(y_{mi}) = \Lambda_m \Lambda_m^\top + \Psi_m, \quad cov(y_{mi}, y_{m'i}) = \Lambda_m A_m^\top A_{m'}\Lambda_{m'}^\top
\end{equation}
and, using results for multivariate Gaussian distributions, for a pair of two views $m \neq m'$, we have
\begin{equation}\label{eq:conditional_distribution}\begin{aligned}
 y_{mi}&\mid y_{m'i} \sim N_{p}(\mu_{i,m \mid m'}, \Sigma_{m \mid m}),\\
    \mu_{i,m \mid m'}& = \Lambda_m A_m^\top A_{m'}\Lambda_{m'}^\top\begin{bmatrix}
        \Lambda_{m'} \Lambda_{m'}^\top + \Psi_{m'}
    \end{bmatrix}^{-1}y_{m'i}, \\
    \Sigma_{m \mid m} &= \Lambda_m \Lambda_m^\top + \Psi_m - \Lambda_m A_m^\top A_{m'}\Lambda_{m'}^\top\begin{bmatrix}
        \Lambda_{m'} \Lambda_{m'}^\top + \Psi_{m'}
    \end{bmatrix}^{-1} \Lambda_{m'} A_{m'}^\top A_{m}\Lambda_{m}^\top.
\end{aligned}   
\end{equation}
The results in \eqref{eq:conditional_distribution} can easily be generalized to more than two views. Equations \eqref{eq:covariances} and \eqref{eq:conditional_distribution} illustrate that the joint modeling of multiple views facilitates straightforward inferences on the covariance structure as well as the prediction of variables in certain views given the others.

\citet{gfa, gfa_2, mofa} adopt a Bayesian approach to fitting model \eqref{eq:mofa} using automatic relevance determination or spike-and-slab priors \citep{spike_slab} to
infer which factors are active in each view. \citet{mofa+, mefisto, mofa_flex} expand \citet{mofa}'s approach, allowing for group structure and spatial correlation across observations. \citet{motl} develop a transfer learning approach to improve inference on small target datasets leveraging larger datasets and \citet{domain} allow for the incorporation of prior information. The posterior is approximated using mean-field variational inference, which is likely to substantially underestimate uncertainty and lacks guarantees on accuracy of the posterior approximation.
\citet{bgfa_sparse} infer the factor loadings with global-local shrinkage priors and compute a point estimate via an Expectation-Maximization algorithm.

Alternatively, Gibbs samplers or other Markov chain Monte Carlo algorithms can be developed for posterior computation under model \eqref{eq:mofa}. However, these sampling algorithms can suffer from slow convergence and mixing even for single-view factor models when data are high dimensional. Additionally, shrinkage priors that are used to identify the number of active factors and in which view each factor is active can be highly sensitive to hyperparameter choice, calling for non-trivial and computationally intensive tuning routines. 

We also highlight another issue with full likelihood inference for model \eqref{eq:mofa}. In the presence of heterogeneity of dimensions and signal-to-noise ratios across data modalities, views with smaller dimension or lower signal-to-noise ratio have a relatively small contribution to the overall data likelihood. The latent factors that act only on these views might be poorly estimated. This issue is exacerbated by methods that employ information criteria or shrinkage priors to estimate the number of latent factors, which might completely drop these latent factors.
The same problem affects most state-of-the-art factor analysis methods, such as \citet{rotate}'s method, applied to the dataset obtained by concatenating the $Y_m$'s. 

In a parallel line of research, \citet{jive, ajive} propose a matrix factorization approach (\texttt{JIVE}) that decomposes each data set into a low-rank joint structure, a low-rank view-specific component, and additive noise. Properties of \texttt{JIVE} have been studied in \citet{yang2025estimating}. The key limitation of \texttt{JIVE} is that it does not provide any quantification of uncertainty.

\citet{fable} recently propose an alternative to Gibbs sampling for Bayesian analysis of high-dimensional single-view factor models. Their approach consists of estimating the latent factors via a matrix factorization and inferring loadings and residual error variances via surrogate regression problems, where the estimates for latent factors are treated as observed covariates. They show posterior
concentration and valid coverage for credible intervals of entries of the covariance. Their approach is related to joint estimation for factor models \citep{chen_jmle, chen_identfiability, lee24, flair, blast}, which has been shown to be consistent in double asymptotic scenarios in which both sample size and data dimension diverge. \citet{basil} expand on this approach to incorporate auxiliary information. 
\citet{fable}'s method can have excellent empirical performance when applied to
the concatenated $Y_m$'s in recovering the overall covariance matrix, but suffers from the aforementioned problems in estimating certain intra- or inter-view covariances in presence of view heterogeneity, as we show in our numerical experiments.  

Motivated by these considerations, we propose Factor Analysis for Multi-view data via spectral Alignment (\texttt{FAMA}). First, \texttt{FAMA} estimates the active latent factors in each view (the matrices of $\{\eta_{mi}\}_{i=1}^n$ for $m=1, \dots, M$) by spectral decompositions treating each view separately. This enables accurate estimation of the latent factors for each view, regardless of signal-to-noise ratio and dimension heterogeneity across views. Next, we use another spectral decomposition to align factor estimates from each view. This step provides an estimate of the matrix of $\eta_i$'s, the latent factors that are active in at least one view, up to an orthogonal transformation. Finally, we infer posterior distributions for the loadings and residual variances via surrogate Bayesian regression tasks using the estimate for the latent factors as observed covariates and employing conjugate normal-inverse gamma priors. These posteriors can be computed analytically in parallel for the variables in each view, and induce posteriors on the inter- and intra-view covariances. With a mild analytic inflation of the posterior variance, these covariance posteriors have asymptotically valid frequentist coverage.

We provide theoretical support for our methodology by showing posterior contraction around the true parameter when both the sample size and view dimensions diverge, a central limit theorem for the point estimator, and a Bernstein-von Mises-type result characterizing the asymptotic behavior of the posterior distribution. We show excellent performance in simulations and apply the methodologies to integrate four omics datasets of cancer cell samples from \citet{linkedomics}.

Our contributions include: (i) a spectral decomposition-based estimate of latent factors that are active in any of the views, (ii) a method to obtain point estimates that is notably faster than many alternatives, (iii) a method to quantify uncertainty without any expensive Markov chain Monte Carlo routine, and (iv) substantial theoretical support for our methodology.

\section{Methodology}\label{sec:method}
\subsection{Notation}
For a matrix $A$, we denote by $||A||$, $||A||_F$ and $||A||_\infty$ its spectral, Frobenius, and entry-wise infinity norm, respectively, and by $s_l(A)$ its $l$-th largest singular value. For a vector $v$, we denote by $||v||$ and $||v||_\infty$ its Euclidean and entry-wise infinity norm, respectively.

\subsection{Spectral estimation of latent factors}\label{subsec:factor_estimation}

In this section, we illustrate how to obtain an estimate of the matrix of latent factors.
First, we rewrite the model in its equivalent matrix form:
\begin{equation}\label{eq:model_matrix}
    Y_m  =  F A_m\Lambda_m^\top  + E_m, \quad (m=1, \dots, M),
\end{equation}
where 
\begin{equation*}\begin{aligned}
    Y_m = \begin{bmatrix}
        y_{m1} ~ \cdots ~  y_{mn}
    \end{bmatrix}^{\top} & \in   \mathbb R^{n \times p_m}, \quad F = \begin{bmatrix}
        \eta_{1}~ \cdots ~ \eta_{n}
    \end{bmatrix}^\top \in  \mathbb R^{n \times k_0}, \\ E_m &=\begin{bmatrix}
            \epsilon_{m1} ~ \cdots ~ \epsilon_{mn}
    \end{bmatrix}^{\top} \in  \mathbb R^{n \times p_m}.
\end{aligned}
\end{equation*}

Here, we consider the latent dimensions, $k_0$ and $\{k_m\}_{m=1}^M$, to be known. Section \ref{subsec:hyperparams} describes an information criterion-based approach to estimate them. First, we compute the singular value decomposition of each view and let $P_m = U_m U_m^{\top}$, where $U_m \in  R^{n \times k_m}$ is the matrix of left singular vectors of $Y_m$ associated to its leading $k_m$ singular values. Proposition \ref{prop:reovery_P_m} in the supplementary material shows that $P_m$ approximates the orthogonal projection matrix projecting on the column space spanned by the true signal $F A_m \Lambda_m^\top$ for large values of $n$ and $p_m$. 
Hence, intuitively, for each pair $m \neq m'$, $P_m$ and $P_{m'}$ approximate projections that span subspaces with a partially shared basis, corresponding to the main axes of variation of latent factors that are shared by views $m$ and $m'$. Moreover, $\sqrt{n} U_m$ coincides with the principal component estimates of latent factors for $Y_m$ \citep{bai_03}.
Next, we let $\tilde P$ be the empirical average of the $P_m$'s:
\begin{equation}\label{eq:P_tilde}
    \tilde P = \frac{1}{M}\sum_{m=1}^M P_m.
\end{equation}
Finally, the estimate of the latent factors is $\hat F = \sqrt{n} U$, where $U \in  R^{n \times k_0}$ is the matrix of left singular vectors of $\tilde P$ associated to its $k_0$ leading singular values. Since the latent factors are independent standard Gaussian random variables, the different columns of $F$ span linearly independent directions. Therefore, for each column of $F$, at least one $P_m$ will project onto a subspace with its basis including the corresponding direction. Hence, the first $k_0$ singular vectors of $\tilde P$ approximately span the same space as the latent factors and their corresponding singular values will approximately be $\geq 1/M$.
Theorem \ref{thm:recovery_factors_overall} in Section \ref{sec:theory} shows that the Procrustes-Frobenius error in $\hat F$ vanishes as $n$ and the $p_m$'s diverge.  This procedure (summarized in Algorithm \ref{alg:factor_estimation}) provides an estimate (up to rotation) of the matrix of the latent factors that are active in at least one view. Hence, we obtain an estimate of the latent factors that are active in at least one view, instead of trying to infer in which view each latent factor is active. This allows us to avoid some of the statistical and computational hurdles associated to such task, while still yielding accurate covariance estimates. In Section \ref{sec:bd} of the Supplementary Material, we discuss a way to recover shared and view-specific latent factors.
\begin{algorithm}
\caption{ Latent factor estimation procedure.
}\label{alg:factor_estimation}\vspace{-2em}
\begin{tabbing}
  Input: The data matrices $\{Y_m\}_{m=1}^M$, %
   and the view-specific latent dimensions $\{k_m\}_{m=1}^M$.\\
   Step 1: For each $m=1, \dots, M$, compute the singular value decomposition of $Y_m$ and let \\
   \quad \enspace \quad \quad $P_m = U_m U_m^\top$, where $U_m \in  R^{p_m \times k_m}$ denotes the matrix of right singular vectors of \\  \quad \enspace \quad \quad $P_m$ corresponding to the $k_m$ leading singular values. \\
   \text{Step 2:} Compute the empirical mean of the $P_m$'s as $\tilde P = \frac{1}{M} \sum_{m=1}^M P_m$.
\\
  \text{Step 3:} Estimate the overall latent dimension $k_0$ using the criterion \eqref{eq:k_0_hat}.
\\

  \text{Step 4:} Estimate the latent factors as $\hat F = \sqrt{n} U$, where $U \in  R^{n \times k_0}$ is the matrix of  left \\ \quad \enspace \quad \quad singular vectors of $\tilde P$ associated to the $k_0$ leading singular values.\\

 \text{Output:} The estimate for the latent factors $\hat F$.
\end{tabbing}
\end{algorithm}

\subsection{Inference of loadings}\label{subsec:inference_Lambda}

After estimating the latent factors, we infer the loadings via the surrogate multilinear regression:
\begin{equation} 
    \hat Y_m = \hat F  \tilde \Lambda_m^\top +\tilde E_m, \quad \text{vec}\big(\tilde E_m \big) \sim N_{np_m}(0, \tilde \Psi_m \otimes I_n),
\end{equation} 
where we introduced new parameters $\tilde \Lambda_m \in  R^{p_m \times k_0}$ and $ \tilde \Psi_m = \text{diag}\big(\tilde \sigma_{m1}^2, \dots, \tilde \sigma_{m p_m}^2\big)$ and treated the estimated latent factors matrix $\hat F$ as the observed covariate matrix. We estimate the factors up to an orthogonal transformation: $\hat F  \approx F R$ for some orthogonal matrix $R$. Ideally, if $\hat F \approx F$, we could estimate $\tilde \Lambda_m$ via a sparsity-inducing prior or penalty and recover $\Lambda_m A_m^\top$ by identifying which columns of $\hat F$ are active in view $m$. However, since $\hat F \approx F R$, we only want $\tilde \Lambda_m$ to be such that $\tilde \Lambda_m  \approx \Lambda_m A_m^\top R$ and to recover intra- and inter-omics covariances through quantities $ \tilde \Sigma_m= \tilde \Lambda_m \tilde \Lambda_m^\top + \tilde \Psi_m$ and $\tilde \Sigma_{mm'} = \tilde \Lambda_m \tilde \Lambda_{m'}^\top$.  
In high dimensions, we have $k_0 < \sum_{m=1}^M k_m \ll p_m$ for $m=1, \dots, M$. Hence, our approach still massively reduces the dimensionality in estimating the covariances to a problem of inferring $p=\sum_{m=1}^M p_m$ $k_0$-dimensional regression problems. We adopt conjugate Normal-Inverse Gamma priors on the rows of $\tilde \Lambda_m$ and residual variances ($\{\tilde \lambda_{mj}, \tilde \sigma_{mj}^2\}_{j=1, m=1}^{p_m, M}$),
\begin{equation}
     \tilde \lambda_{mj} \mid \tilde \sigma_{mj}^2 \sim N_{k_0}\big(0, \tau_m^2 \tilde \sigma_{mj}^2 I_{k_0}\big),  \quad \tilde \sigma_{mj}^2 \sim IG \Big( \frac{\nu_0}{2}, \frac{\nu_0 \sigma_0^2}{2}  \Big),
\quad (j=1, \dots, p_m; m=1, \dots, M).
    \label{eq:prior_Lambda}
      \end{equation} 
The resulting posterior, 
for $j=1, \dots, p_m$ and $m=1, \dots, M$, has the form 
\begin{equation}\begin{aligned}
    (\tilde \lambda_{mj}, \tilde \sigma_{mj}^2) \mid y_m^{(j)}& \sim  NIG \big(\tilde \lambda_{mj}, \tilde \sigma_{mj}^2; \hat \lambda_{mj}, K_m,  \nu_n/2,  \nu_n \delta_{mj}^2/2\big)\\
    &= N_{k_0}\big(\tilde \lambda_{mj}; \hat \lambda_{mj}, \tilde \sigma_{mj}^2 K_m\big)IG \big(\tilde \sigma_{mj}^2;  \nu_n/2,  \nu_n \delta_{mj}^2/2\big),
\end{aligned}
\end{equation}
where 
\begin{equation}\label{eq:posterior_params}
    \begin{aligned}
         \hat \lambda_{mj} &=  \big(\hat F^\top \hat F + \tau_m^{-2} I_{k_0} \big)^{-1}\hat F^{\top} y_m^{(j)} = \frac{1}{n + \tau_\Lambda^{-2}} \hat F^\top y_m^{(j)}, \\
         K_m & = \big(\hat F^\top \hat F + \tau_m^{-2} I_{k_0}\big)^{-1} = \frac{1}{n + \tau_m^{-2}} I_{k_0},\\
         \nu_{n} &= \nu_0 + n,\\
         \delta_{mj}^2 &= \frac{1}{\nu_n} \left(\nu_0 \sigma_0^2  +y_m^{(j)\top} y_m^{(j)} - \hat \lambda_{mj}^\top K^{-1}\hat \lambda_{mj}\right),
    \end{aligned}
\end{equation}
and $y_m^{(j)} $ is the $j$-th column of $Y_m$. The posterior mean for $\Lambda$ is given by
\begin{equation}\label{eq:mu_Lambda}
 \hat \Lambda_m =   \begin{bmatrix}
        \hat \lambda_{m1} & 
        \cdots & 
        \hat \lambda_{mp_m}
    \end{bmatrix}^\top = \frac{1}{n + \tau_m^{-2}} Y_m^\top \hat F,
\end{equation}
which is equivalent to the 
$L_2$-regularized least squares estimate \citep{ridge},
\begin{equation*}
     \hat \Lambda_m  = \underset{ B_m \in \mathbb R^{p_m \times k_0}}{\operatorname{argmin}} || Y_m - \hat F B_m^\top ||_F^2 + \tau_m^{-2} ||B_m||_F^2.
\end{equation*}

Although this procedure produces accurate point estimates, it induces a posterior distribution with slight undercoverage, as formalized in Section \ref{sec:theory}. We solve this problem by inflating the posterior variance of $\tilde \Lambda_m$  by a factor $\rho_m^2 > 1$. Section \ref{subsec:hyperparams} describes a tuning procedure for the $\rho_m$'s, which guarantees asymptotically valid frequentist coverage. The resulting coverage-corrected posterior takes the following form:
\begin{equation}\label{eq:lambda_sigma_cc}\begin{aligned}
    (\tilde \lambda_{mj}, \tilde \sigma_{mj}^2) \mid y_m^{(j)} \sim & NIG \big(\tilde \lambda_{mj}, \tilde \sigma_{mj}^2; \hat \lambda_{mj}, \rho_m^2 K, \nu_n/2, \nu_n \delta_{mj}^2/2\big)\\
    &= N_{k_0}\big(\tilde \lambda_{mj}; \hat \lambda_{mj},  \rho_m^2 \tilde  \sigma_{mj}^2 K\big)IG \big(\tilde \sigma_{mj}^2; \nu_n/2, \nu_n \delta_{mj}^2/2\big).
\end{aligned}
\end{equation}

\subsection{Hyperparameter selection }\label{subsec:hyperparams}
This section describes how to choose hyperparameters of \texttt{FAMA}. First, we propose an information criterion-based approach for estimating latent dimensions $k_m,m=1,\ldots,M,k_0$. We estimate $k_m$ by minimizing the joint likelihood-based information criterion (JIC) of \citet{chen_jic},
\begin{equation*}\label{eq:JIC}
    \text{JIC}_m(k) = -2 l_{mk} + k \max(n,p_m) \log \{\min(n,p_m)\},
\end{equation*}
where $l_{mk}$ is the log-likelihood of the data in the $m$-th view computed at the joint maximum likelihood estimate when the latent dimension is equal to $k$. Calculating the joint maximum likelihood estimate for each value of $k$ can be computationally expensive. Therefore, we approximate $l_{mk}$ with $ l_{mk} \approx \hat l_{mk} $, where $\hat l_{mk}$
 is the likelihood obtained by estimating the latent factors as the leading left singular values of $Y_m$ scaled by $\sqrt{n}$ and the factor loadings by their conditional mean given such an estimate. More details are provided in the supplementary material. Thus, we set
\begin{equation}\label{eq:k_hat}
    \hat k_m = \arg \min_{k_m=1, \dots, k_{m, \max}} \widehat{ \text{JIC}}_m(k_m), \quad \widehat{ \text{JIC}}_m(k)= -2 \hat l_{mk} + k \max(n,p_m) \log\{\min(n,p_m)\},
\end{equation}
where $k_{m, max}$ is an upper bound on $k_m$. Next, $k_0$ is estimated as 
\begin{equation}\label{eq:k_0_hat}
    \hat k_0 = \arg \max_{j=\min\{\hat k_1, \dots, \hat k_S\}, \dots, k_{0, \max}} s_{j}(\tilde P) - s_{j+1}(\tilde P)
    , \quad \text{such that } s_{j+1}(\tilde P) < 1/M - \tau_1, 
\end{equation}
where $k_{0, max}$ is an upper bound on $k_0$ and $\tau_1$ is a small value, which is set to
$1/(2M)$ in practice. The rationale for \eqref{eq:k_0_hat} is as follows. Since latent factors are independent standard Gaussian random variables, the columns of $F$ are approximately orthogonal. Hence, the singular vectors of $\tilde P$ associated with singular values larger than $1/M - \tau_1$ correspond to latent factors that are active in at least one view, while the others have minimal impact and can be removed.  

To tune the variance inflation terms, we take inspiration from \citet{fable} and, for $j,j' = 1, \dots, p_m$, we let \begin{equation}\label{eq:b_ms}
    b_{mjj'} = \begin{cases}
        \big( 1 +  \frac{||\hat \lambda_{mj}||_2^2 ||\hat \lambda_{mj'} ||_2^2 + (\hat \lambda_{mj}^\top \hat \lambda_{mj'})^2 }{ \delta_{mj}^2 ||\hat \lambda_{mj'} ||_2^2 + \delta_{mj'}^2||\hat \lambda_{mj} ||_2^2} \big)^{1/2}, \quad &\text{if } j \neq j'   \\   \big(1 + \frac{ ||\hat \lambda_{mj} ||_2^2 }{ 2\delta_{mj}^2 }\big)^{1/2}, \quad &\text{otherwise.} 
              \end{cases}
\end{equation}
The choice $\rho_m = \max_{1 \leq j \leq j' \leq p_m} b_{mjj'}$ guarantees valid frequentist coverage of entry-wise credible intervals. In practice, we chose $\rho_m = {p_m \choose 2}^{-1} \sum_{1 \leq j \leq j' \leq p_m} b_{mjj'}$ which allows control of the coverage on average across entries. We justify this choice in Section \ref{sec:theory}.

Finally, we propose a data-adaptive strategy for the prior variances. Conditionally on $\{\tilde \sigma_{mj}^2\}_{j=1}^{p_m}$, the \textit{a priori} expected value of the squared Frobenius norm of $\tilde \Lambda_m$ is $E[|| \tilde \Lambda_m||_F^2 \mid \tau_m^2, \{\tilde \sigma_{mj}^2\}_{j=1}^{p_m} ] = \tau_m^2k_m \sum_{j=1}^{p_m}\tilde \sigma_{mj}^2$. Note that $L_m = || U_m^\top Y_m||_F^2 /n$ and $\hat \sigma_m^2 = \sum_{j=1}^{p_m} \hat \sigma_{jm}^2$, where $\hat \sigma_{mj}^2 = ||(I_n - U_m^c U_m^{c\top}) y_m^{(j)}||_2^2/n$, are consistent estimates for $||\Lambda||_F^2$ and $\sum_{j=1}^{p_m} \sigma_{mj}^2$. Thus, we set $\tau_m^2$ to $\hat \tau_m^2 = \frac{L_m}{k_m \hat \sigma_m^2}$.

\subsection{Overall algorithm and inference}

The \texttt{FAMA} procedure is summarized in Algorithm \ref{alg:fama}. This algorithm generates samples from all the model components, including intra- and inter-view covariances and residual variances. These samples are drawn independently, substantially improving the computational efficiency over MCMC. Similarly to MCMC, posterior summaries, including point and interval estimates, can be calculated for any functional of the model parameters. In addition, conditionally on these parameters, one can impute missing or held-out values of the data.

\begin{algorithm}
\caption{\texttt{FAMA} procedure to obtain $N_{MC}$ approximate posterior samples.}\label{alg:fama}\vspace{-2em}
\begin{tabbing}
    Input: The data matrices $ \{Y_{m}\}_{m=1}^M$, the number of Monte Carlo samples $N_{MC}$, and an upper\\ \qquad \enspace \hspace{0.3em} bound on the number of factors $\{k_{m, \max}\}_{m=1}^M, k_{0, \max}$. \\
   Step 1: For each $m=1, \dots, M$, estimate the number of latent factors for each view $k_m$ \\ \qquad \enspace \hspace{0.3em} via equation \eqref{eq:k_hat}.\\
   
    \text{Step 2:} Obtain the estimates for the latent factors $\hat  F_s$  using Algorithm \ref{alg:factor_estimation}. \\
  \text{Step 3:} For each $m=1, \dots, M$, compute the variance inflation term for $\tilde \Lambda_m$, as  \\ \qquad \enspace \hspace{0.3em} $\rho_m = \frac{1}{{p_m \choose 2}} \sum_{1 \leq j \leq j' \leq p_m} b_{mjj'}$ where the $b_{mjj'}$'s are defined in \eqref{eq:b_ms}.\\ 
  \text{Step 4:} For each $m=1, \dots, M$, estimate the hyperparameters $\tau_m$ as described in Section \ref{subsec:hyperparams}.\\
 \text{Step 5:} For each $m=1, \dots, M$, estimate the mean for $ \tilde \Lambda_m$, $\hat \Lambda_m$, as in \eqref{eq:mu_Lambda} and, for each  $j=1, \dots, p_m$  \\ \qquad \enspace \hspace{0.3em} in parallel, for $t = 1, \dots, N_{MC}$, sample independently $(  \tilde \lambda_{mj}^{(t)},  \tilde \sigma_{mj}^{2(t)})$ from \eqref{eq:lambda_sigma_cc}.  \\
 \text{Output:} $N_{MC}$ samples of the intra-view covariance low-rank components and residual variances \\ \qquad \enspace \hspace{0.3em}  $\{ \tilde \Lambda_m^{(1)} \tilde \Lambda_m^{(1)\top}, \ldots,  \tilde \Lambda_m^{(N_{MC})}\tilde \Lambda_m^{(N_{MC})\top}\}_{m=1}^M$,  $\{\{ \tilde \sigma_{mj}^{2 (1)}\}_{j=1}^{p_m}, \dots,\{ \tilde \sigma_{mj}^{2 (N_{MC})}\}_{j=1}^{p_m} \}_{m=1}^M$,  \\ \qquad \enspace \hspace{0.3em} and inter-view covariances $\{ \tilde \Lambda_m^{(1)} \tilde \Lambda_{m'}^{(1)\top}, \ldots,  \tilde \Lambda_m^{(N_{MC})}\tilde \Lambda_{m'}^{(N_{MC})\top}\}_{1 \leq m \neq m' \leq M}$.

\end{tabbing}
\end{algorithm}

\section{Theoretical support}\label{sec:theory}
\subsection{Preliminaries}\label{subsec:prelim}
This section provides theoretical support for our methodology. We provide favorable results in the double asymptotic regime, that is when $n$ and the $p_m$'s diverge. In this sense, our method enjoys a blessing of dimensionality, with its accuracy increasing as the dimensionalities (and the sample size) grow. 
Before stating our results, we enumerate some regularity conditions.

\begin{assumption}[Model]\label{assumption:model} The data are generated from model \eqref{eq:mofa}, and the true loading matrix and residual error variances for the $m$-th view are denoted by $\Lambda_{0m}$ and $\{\sigma_{0mj}^2\}_{j=1}^{p_m}$ respectively. We denote the $j$-th row of $\Lambda_{0m}$ by $\lambda_{0mj}$. We denote the true latent factors by $F_0$.
\end{assumption}

\begin{assumption}[Dimensions and sample size]\label{assumption:dimensions}
    We have $p_1, \dots, p_M \to \infty$ with $(\log p_m) / n =  o(1)$ for $m=1, \dots, M$.
\end{assumption}

\begin{assumption}[Idiosyncratic variances]\label{assumption:sigma}
The idiosyncratic error variances are such that $0< c_{\sigma 1} \leq \sigma_{0mj}^2 \leq c_{\sigma 2} < \infty$ for $j=1, \dots, p_m$ and $m=1, \dots, M$.   
\end{assumption}

\begin{assumption}[Loadings]\label{assumption:Lambda}
The true loading matrices are such that $s_l(\Lambda_{0m}) \asymp \sqrt{p_m}$ for $l=1, \dots, k_m$, $||\Lambda_{0m}||_\infty \leq  \infty$ and $\min_{j=1, \dots, p_m } ||\lambda_{0mj}||_2 \geq c_{\Lambda} >0$ for $m=1, \dots, M$.    
\end{assumption}

\begin{assumption}[Hyperparameters]\label{assumption:hyperparameters}
    The hyperparameters $\nu_0, \delta_0, k_0, k_1, \dots, k_M, \rho_1, \dots, \rho_M$, $ \tau_1, \dots, \tau_M$ are fixed constants.
\end{assumption}

\subsection{Accuracy of latent factor estimates}\label{subsec:accuracy_factors}
The first result quantifies the accuracy of the overall factor estimates.
\begin{theorem}[Recovery of the overall latent factors]\label{thm:recovery_factors_overall}
   Suppose Assumptions \ref{assumption:model}--\ref{assumption:hyperparameters} hold,  
   then, as $n, p_1, \dots, p_M, \to \infty$,
    with probability at least $1-o(1)$,
    \begin{equation*}
        \begin{aligned}
           \min_{R \in  \mathbb R^{k_0 \times k_0} : R^\top R = I_{k_0}} \frac{1}{\sqrt{n}}||\hat F R- F_0||& \lesssim \frac{1}{\sqrt{n}} + \frac{1}{p_{\min}}.
        \end{aligned}
    \end{equation*}
where $p_{\min} = \min_{m=1, \dots M}p_m$.
\end{theorem}
Theorem \ref{thm:recovery_factors_overall} states that the Procrustes Frobenius error for the matrix of all latent factors is small with high probability as $n$ and the $p_m$'s diverge. We consider the Procrustes error, minimizing over the set of orthogonal transformations, because the parameters in the factor models are only identifiable up to rotation without additional restrictions. This result justifies the use of our point estimate for the latent factors as a useful low-dimensional summary of the high-dimensional data.

\subsection{Asymptotics for covariances estimates}\label{subsec:accuracy_loadings}
Next, we show results on estimates of parameters in the model \eqref{eq:mofa}.  
\begin{theorem}\label{thm:posterior_contraction}
    If Assumptions \ref{assumption:model}--\ref{assumption:hyperparameters} hold, then, as $n, p_1, \dots, p_M \to \infty$, there exist finite constants $\{D_m\}_{m=1}^M$ such that
    \begin{equation}\label{eq:posterior_contraction_intra}
    \begin{aligned}
        E\left[\tilde \Pi\left\{\frac{||\tilde \Lambda_m \tilde \Lambda_m^\top - \Lambda_{0m}  \Lambda_{0m}^\top||}{||\Lambda_{0m}  \Lambda_{0m}^\top||} > D_m \left( \frac{1}{\sqrt{n}} + \frac{1}{\sqrt{p_m}} + \frac{1}{p_{\min}} \right) \right\} \right] &\to 0,\quad (m=1, \dots, M).
    \end{aligned}
    \end{equation}
Moreover, for all $m, m' \in \{1, \dots, M\}$, with $m \neq m'$, such that $||\Lambda_{0m} A_m^\top A_{m'} \Lambda_{0m'}^\top || \asymp \sqrt{p_m p_{m'}}$, there exist finite constants $\{D_{mm'}\}_{1 \leq m \neq m' \leq M}$, such that 
 \begin{equation}\label{eq:posterior_contraction_inter}
    \begin{aligned}
        E\left[\tilde \Pi\left\{\frac{||\hat \Lambda_m\hat \Lambda_{m'}^\top  - \Lambda_{0m} A_m^\top A_{m'} \Lambda_{0m'}^\top||}{||\Lambda_{0m} A_m^\top A_{m'} \Lambda_{0m'}^\top||} > D_{mm'} \left( \frac{1}{\sqrt{n}} + \frac{1}{\sqrt{p_m}} + \frac{1}{\sqrt{p_{m'}}} + \frac{1}{p_{\min}} \right) \right\} \right] &\to 0,
    \end{aligned}
    \end{equation}
In particular, for the low-rank component of the intra-view covariance, we have, 
 with probability at least $1-o(1)$, 
 \begin{equation*}
 \frac{||\hat \Lambda_m \hat \Lambda_m^\top  - \Lambda_{0m}  \Lambda_{0m}^\top||}{||\Lambda_{0m}  \Lambda_{0m}^\top||} \lesssim \frac{1}{\sqrt{n}} + \frac{1}{\sqrt{p_m}} + \frac{1}{p_{\min}},
\end{equation*}
and for the inter-view covariances (with $ ||\Lambda_{0m}  A_m^\top A_{m'}\Lambda_{0m'}^\top || \asymp \sqrt{p_m p_{m'}}$),  we have 
\begin{equation*}
 \frac{||\hat \Lambda_m\hat \Lambda_{m'}^\top  - \Lambda_{0m}  A_m^\top A_{m'} \Lambda_{0m'}^\top||}{||\Lambda_{0m}  A_m^\top A_{m'}\Lambda_{0m'}^\top||} \lesssim \frac{1}{\sqrt{n}} + \frac{1}{\sqrt{p_m}} + \frac{1}{\sqrt{p_{m'}}} + \frac{1}{p_{\min}}.
\end{equation*}
\end{theorem}
Theorem \ref{thm:posterior_contraction} bounds the relative error in the operator norm with high probability. We rescale the error by dividing by the norm of the true parameters, so the results are comparable as dimensionality grows. Equations \eqref{eq:posterior_contraction_intra} and 
\eqref{eq:posterior_contraction_inter} characterize the behavior of the posterior distribution on the low-rank components of the covariance, providing the contraction rate around the true values. The second part of the theorem supports the use of $\hat \Lambda_m \hat \Lambda_m^\top$ and $\hat \Lambda_m \hat \Lambda_{m'}^\top$ as point estimates for the low-rank components of the intra- and inter-omics covariance. The next result characterizes the asymptotic point-wise distribution for the entries of $\hat \Lambda_m \hat \Lambda_m^\top$ and $\hat \Lambda_m \hat \Lambda_{m'}^\top$.
\begin{theorem}[Central limit theorem]\label{thm:clt}
      Suppose Assumptions \ref{assumption:model}--\ref{assumption:hyperparameters} hold and $n/\sqrt{p_{\min}} = o(1)$. For $m=1, \dots, M$ and $j,j' = 1, \dots, p_m$, let $S_{0mjj'}$ be
       \begin{equation}\label{eq:S_0_sq}
              S_{0mm'jj'}^2 =    \begin{cases}
                 4\sigma_{0mj'}^2||\lambda_{0mj}||^2 + 2||\lambda_{0mj'}||^4 , \quad & \text{if } j=j' \text{ and } m=m',\\
                 \begin{aligned}
                     &\sigma_{0mj'}^2||\lambda_{0mj}||^2   + \sigma_{0mj}^2||\lambda_{0mj'}||^2 + (\lambda_{0mj}^\top \lambda_{0mj'})^2 \\
                     &+ ||\lambda_{0mj}||^2 ||\lambda_{0mj'}||^2, \\
                 \end{aligned}
              \quad & \text{if } j \neq j'  \text{ and } m=m',\\
              \begin{aligned}
             & \sigma_{0m'j'}^2||\lambda_{0mj}||^2   + \sigma_{0mj}^2||\lambda_{0m'j'}||^2 + (\lambda_{0mj}^\top A_m^\top A_{m'} \lambda_{0m'j'})^2 \\  
             &+  || \lambda_{0mj}||^2 || \lambda_{0m'j'}||^2,     
             \end{aligned}
 \quad & \text{otherwise}.
             \end{cases}
       \end{equation} 
       Then, as $n, p_1, \dots, p_M \to \infty$, for all $m, m'=1, \dots, M$,
       \begin{equation*}
           \begin{aligned}
             & \frac{\sqrt{n}}{S_{0mm'jj'}} \big(\hat \lambda_{mj}^\top \hat \lambda_{m'j'} - \lambda_{0mj}^\top A_m^\top A_{m'}\lambda_{0m'j'} \big) \Rightarrow N(0, 1),(j=1, \dots, p_m; j'=1, \dots, p_{m'}).
           \end{aligned} 
       \end{equation*}     
\end{theorem}
Theorem \ref{thm:clt} shows that rescaled and centered entries of
$\hat \Lambda_m \hat \Lambda_m^\top$ are normally distributed asymptotically, that is, for large $n$ and $p_m$'s, $\hat \lambda_{mj}^\top \hat \lambda_{m'j'} \overset{\cdot}{\sim} N( \lambda_{0mj}^\top \lambda_{0m'j'}, \frac{S_{0mm'jj}^2}{n})$ where 
$S_{0mm'jj'}$ depend on unknown parameters that can be consistently estimated. In practice, we replace $S_{0mm'jj'}$ with $\hat S_{mm'jj'}$, which is obtained by using the $\hat \lambda_{mj}$'s and $\delta_{mj}$'s, which are defined in \eqref{eq:posterior_params}, instead of the true values of the parameters. In the supplemental, we show that $\hat S_{mm'jj'}$ is a consistent estimate for $S_{0mm'jj'}$. Hence, an asymptotically valid $(1-\alpha) \%$ confidence interval for $ \lambda_{0mj}^\top \lambda_{0m'j'}$ is 
\begin{equation}\label{eq:conf_int}
    \begin{aligned}
        \hat C_{mm'jj'} = \bigg[\hat \lambda_{mj}^\top \hat \lambda_{m'j'} \pm z_{1-\alpha/2}\frac{\hat S_{mjj'}}{\sqrt{n}}\bigg],
    \end{aligned}
\end{equation}
where $z_{q} = \Phi^{-1}(q)$ and $\Phi(\cdot)$ denotes the cumulative distribution function of a standard Gaussian random variable. 

\begin{theorem}[Bernstein von Mises theorem]\label{thm:bvm}
    Suppose Assumptions \ref{assumption:model}--\ref{assumption:hyperparameters}  hold and $n/\sqrt{p_{\min}} = o(1)$. For $m, m' = 1, \dots, M$, $j=1, \dots, p_m$ and $j' = 1, \dots, p_{m'}$, let
    \begin{equation}\label{eq:T_0_sq}
        T^2_{0mm'jj'}(\rho_m, \rho_{m'})  = \begin{cases}
4\rho_m^2 \sigma_{0mj}^2 || \lambda_{0mj}||^2, \quad &\text{if } j=j' \text{ and } m=m',\\
  \rho_{m'}^2 \sigma_{0m'j'}^2 || \lambda_{0mj}||^2 +  \rho_m^2 \sigma_{0mj}^2 || \lambda_{0m'j'}||^2 , \quad &  \text{otherwise}.
\end{cases}
    \end{equation}
    Then, as $n, p_1, \dots, p_M \to \infty$, we have 
    \begin{equation}
        \sup_{x \in \mathbb R} \Bigg| \tilde \Pi \bigg\{\frac{\sqrt{n}}{T_{0mm'jj'}(\rho_m, \rho_{m'})}\big(\tilde \lambda_{mj}^\top \tilde \lambda_{m'j'} - \hat \lambda_{mj}^\top \hat \lambda_{m'j'}\big) \leq x \bigg\} - \Phi(x)\Bigg| \overset{pr}{\to} 0,   \end{equation}
        for all $ j = 1, \dots, p_m, j' = 1, \dots, p_{m'}$ and $m,m' =1, \dots, M$.
\end{theorem}
Theorem \ref{thm:bvm} shows that
the posterior distribution of centered and rescaled entries of the low-rank components converges to a zero-mean Gaussian distribution with variance $T_{0mm'jj'}(\rho_m, \rho_{m'})$. 
This result leads to approximate $(1-\alpha)\%$ credible intervals for large values of $n$ and $p_m$'s as  
\begin{equation}\label{eq:cred_int_approx}
    \begin{aligned}
        \tilde C_{mm'jj'}(\rho_m, \rho_{m'}) = \bigg[\hat \lambda_{mj}^\top \hat \lambda_{m'j'} \pm z_{1-\alpha/2}\frac{ \hat T_{mm'jj'}(\rho_m, \rho_{m'})}{\sqrt{n}}\bigg],
    \end{aligned}
\end{equation}
where we use a consistent estimator $\hat T_{mm'jj'}(\rho_m, \rho_{m'})$ of $T_{0mm'jj'}(\rho_m, \rho_{m'})$, obtained by replacing $\lambda_{0mj}$'s and $\sigma_{0mj}^2$'s with estimates $\hat \lambda_{mj}$'s and $\delta_{mj}^2$'s. The term $T_{0mm'jj'}(\rho_m, \rho_{m'})$ is needed to calibrate the credible intervals to the target frequentist level.  
Combining Theorems \ref{thm:clt} and \ref{thm:bvm}, and again plugging in consistent point estimates for unknown parameters, justifies setting $\rho_m =  b_{0mjj'}$, with $b_{0mjj'}$ defined in (\ref{eq:b_ms}), to achieve asymptotically valid coverage for the $j,j'$-th entry of the intra-view low-rank components. Setting $\rho_m $ to the empirical averages of the $b_{0mjj'}$'s allows approximately valid coverage across entries of the intra-view low-rank components. In Section
\ref{sec:numerical_experiments},
we provide numerical evidence that this approach also provides calibrated coverage on average for entries of intra-view covariances.

\section{Numerical experiments}\label{sec:numerical_experiments}

We compare \texttt{FAMA} with Multi-Omics Factor Analysis (\citet{mofa, mofa+}, \texttt{MOFA}), the angle-based joint and individual variance explained method (\texttt{AJIVE}, \citet{ajive}), a spectral estimator based on the singular value decompositions of each data view (\texttt{SVD}), which we describe in the supplemental, two factor analysis methods applied to the concatenated dataset, namely \citet{fable}'s \texttt{FABLE}
and \citet{rotate}'s \texttt{ROTATE}, and the empirical covariance matrix (\texttt{EMPIRICAL}). Additional details are provided in the supplemental. We take $M=4$ views, set the total number of factors $k_0 = 30$, and let $k_m = 20$ for all $m=1, \dots, 4$. The loadings and residual error variances are generated as
$\lambda_{0mj} \sim N_{k_{m}}(0, \psi_{m}^2)$ and $\sigma_{0mj}^2 \sim U(5, 10)$, $j=1, \dots, p_m; m=1, \dots 4.$

We consider an unbalanced scenario with $(p_1, \dots, p_4) = (5000, 1000, 1000, 1000)$ and $(\psi_1, \dots \psi_4) = (1, 0.4, 0.4, 0.4)$ and a balanced scenario with $p_m = 2000$ and $\psi_m = 0.5$ for all $m=1, \dots, 4$. For each scenario, we let $n \in \{250, 500, 1000\}$, and for any configuration, we replicate the experiments $50$ times, each time generating the data from the model \eqref{eq:mofa}.

We evaluated estimation accuracy via the Frobenius error rescaled by the Frobenius norm. We focus on the total, intra-view, and inter-view covariances. We evaluated the accuracy of uncertainty quantification via the frequentist average coverage of entrywise 95\% credible or confidence intervals intervals. For \texttt{FAMA}, we report the results using intervals
derived from the central limit theorem (Theorem \ref{thm:clt}) defined in \eqref{eq:conf_int} and
the Bernstein von Mises theorem (Theorem \ref{thm:bvm}) defined in \eqref{eq:cred_int_approx}. \texttt{AJIVE}, \texttt{MOFA}, and \texttt{ROTATE} only provide point estimates.

\begin{figure}
    \centering
    \includegraphics[width=\linewidth]{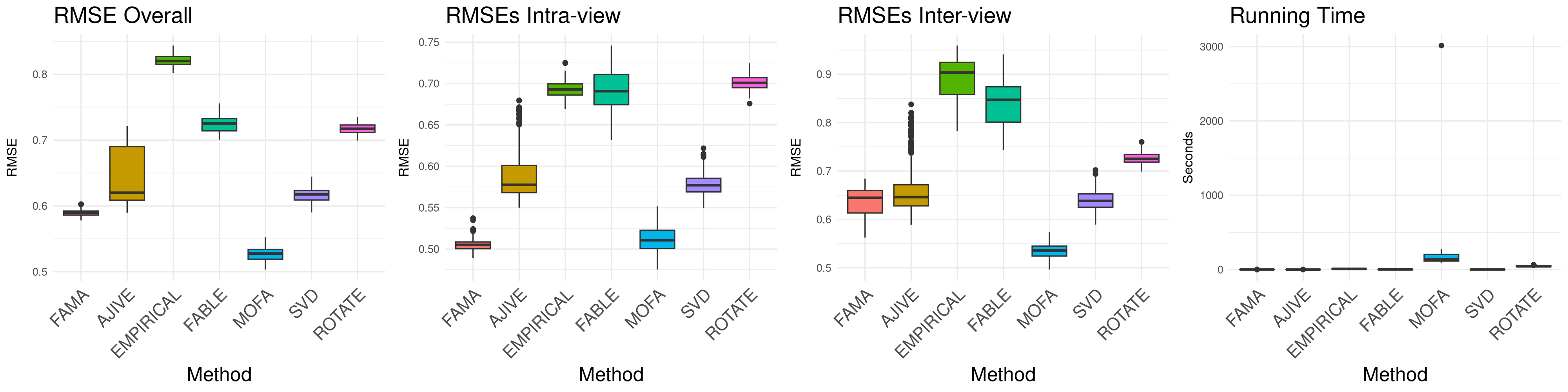}
    \vfill
    \includegraphics[width=\linewidth]{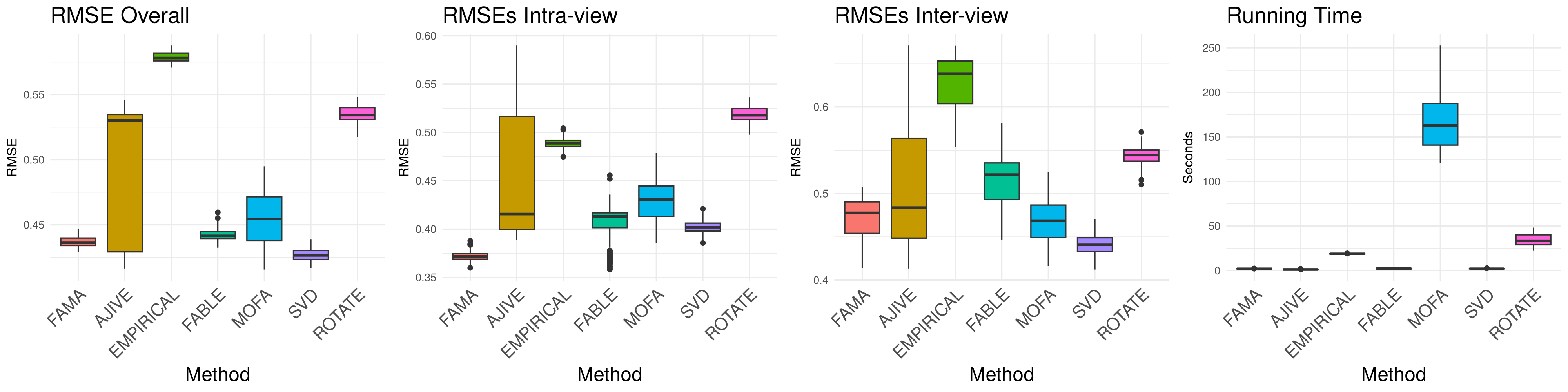}
    \vfill
     \includegraphics[width=\linewidth]{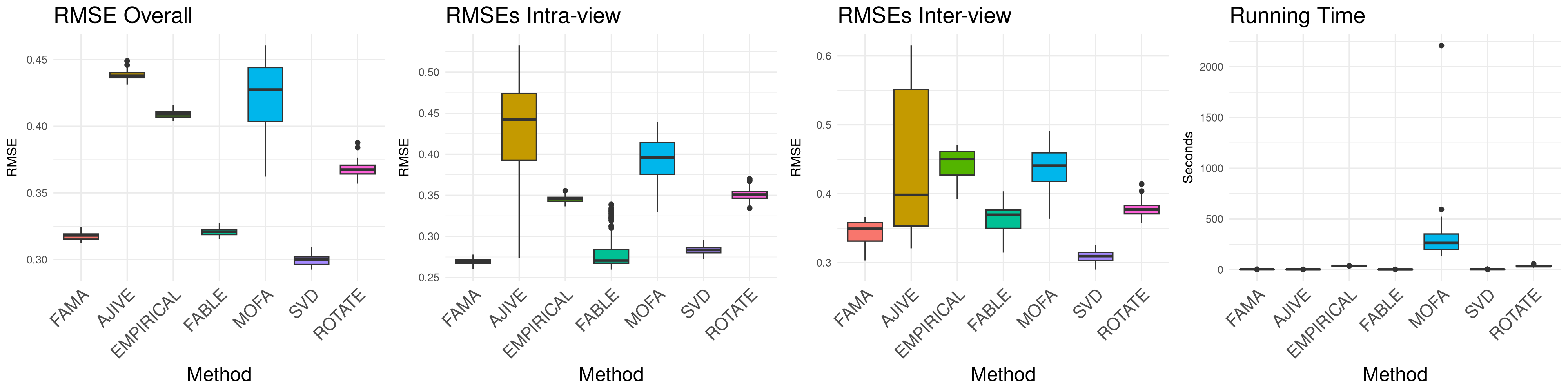}
    \caption{Comparison of different methods in terms of normalized root mean squared error of the overall covariance (left most panels), of the intra-view and inter-view covariances (middle panels) and running time (right most panel) in the balanced scenario with $n=250$ (top panels), $n=500$ (middle panels) and  $n=1000$ (bottom panels).}
    \label{fig:1}
\end{figure}

\begin{figure}
    \centering
    \includegraphics[width=\linewidth]{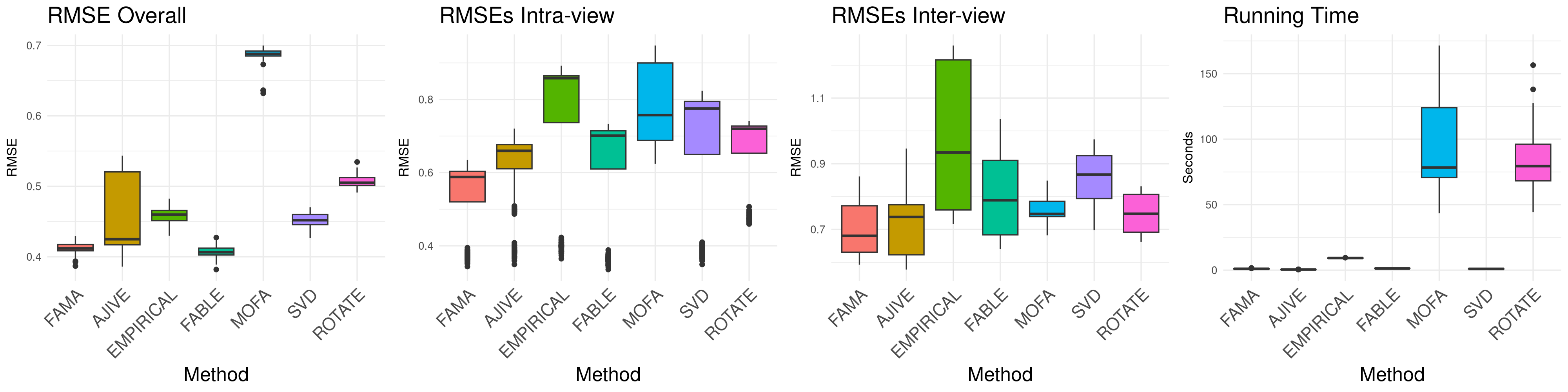}
    \vfill
    \includegraphics[width=\linewidth]{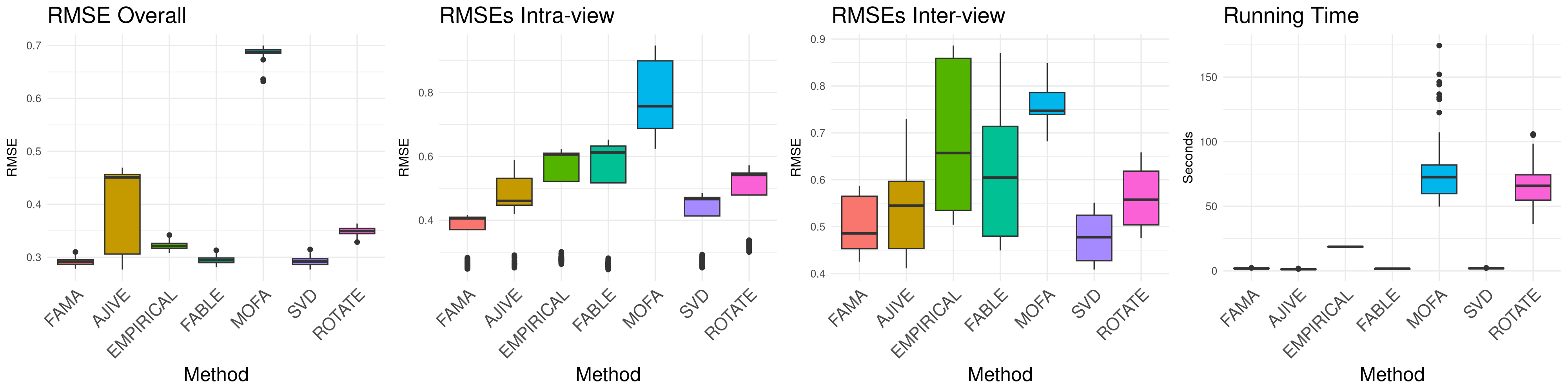}
    \vfill
    \includegraphics[width=\linewidth]{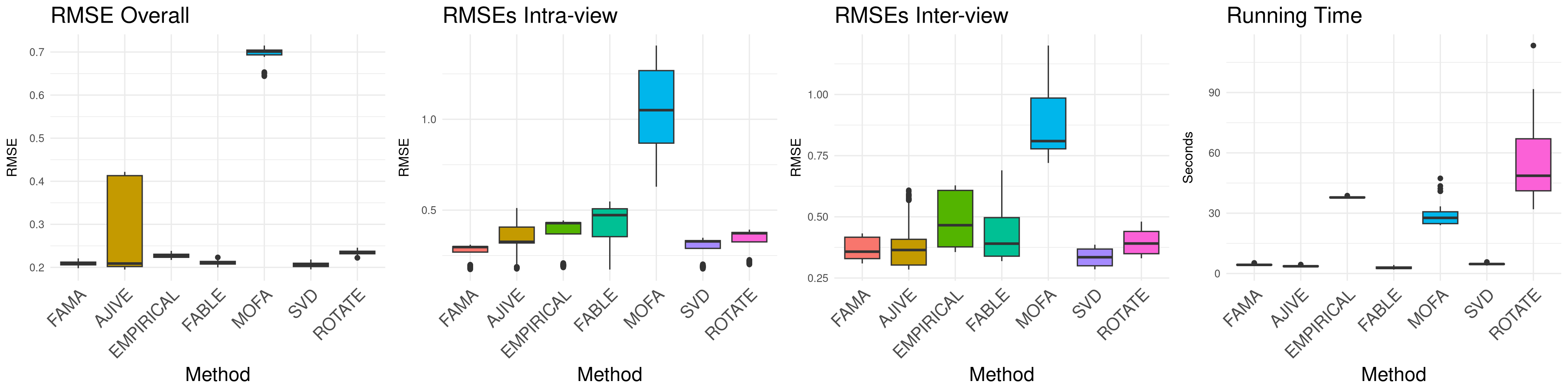}
    \caption{Comparison of different methods in terms of normalized root mean squared error of the overall covariance (left most panels), of the intra-view and inter-view covariances (middle panels) and running time (right most panel) in the unbalanced scenario with $n=250$ (top panels), $n=500$ (middle panels) and  $n=1000$ (bottom panels).}
    \label{fig:2}
\end{figure}

Figures \ref{fig:1} and \ref{fig:2} show the relative Frobenius error of competing methodologies in estimating the covariance (three left most panels) and the running time (right most panel) in the balanced and unbalanced cases, respectively. The left most panel displays the error in estimation of the overall covariance, while the two center panels report the error in estimating the intra- and inter-view covariances, respectively. 

In the balanced case (Figure \ref{fig:1}), when the sample size is small ($n=250$), \texttt{MOFA} is the best performing method, followed by \texttt{FAMA}. With higher sample size, \texttt{FAMA} tends to outperform competitors in recovering the intra-view covariances and \texttt{SVD} has slightly better accuracy for the inter-view covariances. 
In the unbalanced scenario (Figure \ref{fig:2}),  \texttt{FAMA} and \texttt{FABLE} obtain the best performance in recovering the overall covariance with \texttt{FAMA} doing slightly better. However, \texttt{FAMA} remarkably outperforms all other methods in estimating the inter-view covariances. 
As for the intra-view covariances, \texttt{FAMA} and \texttt{AJIVE} have the best performance with small sample sizes, while \texttt{SVD} is the best performing method and is closely followed by \texttt{FAMA}, \texttt{AJIVE}, and \texttt{ROTATE} in the other configurations. In all experiments, \texttt{FAMA}, \texttt{AJIVE}, \texttt{FABLE}, and \texttt{SVD} are the fastest methods with advantages of several orders of magnitude over other methodologies.

Table \ref{tab:uq} reports the average frequentist coverage across entries of the low-rank components of the covariance. Both
types of \texttt{FAMA} intervals yield precise coverage on average across entries of low-rank components of the intra-view and inter-view covariances. \texttt{FABLE}'s interval estimates have comparable length on average, but suffer from a undercoverage in many scenarios.

In summary, \texttt{FAMA} provides accurate estimates with valid uncertainty quantification in all settings considered, while having fast running times.

\begin{table}
	\caption{Average frequentist coverage for entries of random $100\times 100$ submatrices of each $\Lambda_m \Lambda_m^\top$ for all $m$ and $\Lambda_m \Lambda_{m'}^\top$ for all pairs $m \neq m'$. For \texttt{FAMA}, we report both the coverage of confidence intervals obtained from the central limit theorem (Theorem \ref{thm:clt}) defined in \eqref{eq:conf_int}, and of the analytic approximation of the credible intervals (Theorem \ref{thm:bvm}), defined in \eqref{eq:cred_int_approx}, which we refer to as  \texttt{FAMA}-CLT and \texttt{FAMA}-$\tilde \Pi$ respectively. All values have been multiplied by $10^2$. 	\label{tab:uq}
}
    \centering
{\begin{tabular}{c cc c cc c cc}
& \multicolumn{8}{c}{Balanced scenario} \\
& \multicolumn{8}{c}{$n=250$} \\
& \multicolumn{2}{c}{\texttt{FAMA}-CLT} 
& 
& \multicolumn{2}{c}{\texttt{FAMA}-$\tilde \Pi$} 
& 
& \multicolumn{2}{c}{\texttt{FABLE}} \\
& coverage ($\%$) & length & & coverage ($\%$) & length & & coverage ($\%$) & length \\
$\Lambda_m \Lambda_m^\top$    & 95.06 & 2.30 & & 95.06 & 2.30  & & 76.48 & 1.93 \\
$\Lambda_m \Lambda_{m'}^\top$ & 96.56  & 2.39 & & 96.56  & 2.39  & &  84.05 & 1.94 \\
& \multicolumn{8}{c}{$n=500$} \\
& \multicolumn{2}{c}{\texttt{FAMA}-CLT} 
& 
& \multicolumn{2}{c}{\texttt{FAMA}-$\tilde \Pi$} 
& 
& \multicolumn{2}{c}{\texttt{FABLE}} \\
& coverage ($\%$) & length & & coverage ($\%$) & length & & coverage ($\%$) & length \\
$\Lambda_m \Lambda_m^\top$    & 94.51 & 1.64 & & 94.51 & 1.64  & & 90.04  & 1.65 \\
$\Lambda_m \Lambda_{m'}^\top$ & 95.83 & 1.71 & & 95.83 & 1.71 & & 93.81 & 1.65\\
& \multicolumn{8}{c}{$n=1000$} \\
& \multicolumn{2}{c}{\texttt{FAMA}-CLT} 
& 
& \multicolumn{2}{c}{\texttt{FAMA}-$\tilde \Pi$} 
& 
& \multicolumn{2}{c}{\texttt{FABLE}} \\
& coverage ($\%$) & length & & coverage ($\%$) & length & & coverage ($\%$) & length \\
$\Lambda_m \Lambda_m^\top$    & 94.24 & 1.17 & & 94.24 & 1.17   & & 85.44 & 1.17 \\
$\Lambda_m \Lambda_{m'}^\top$ & 95.45 & 1.22 & & 95.46 & 1.22 & & 91.54 & 1.17  \\
& \multicolumn{8}{c}{Unbalanced scenario} \\
& \multicolumn{8}{c}{$n=250$} \\
& \multicolumn{2}{c}{\texttt{FAMA}-CLT} 
& 
& \multicolumn{2}{c}{\texttt{FAMA}-$\tilde \Pi$} 
& 
& \multicolumn{2}{c}{\texttt{FABLE}} \\
& coverage ($\%$) & length & & coverage ($\%$) & length & & coverage ($\%$) & length \\
$\Lambda_m \Lambda_m^\top$    & 92.78 & 2.87  & & 92.78 & 2.87  & &  93.50 & 3.07   \\
$\Lambda_m \Lambda_{m'}^\top$ & 95.18& 2.71  & & 95.18& 2.71  & &   97.09& 3.27 \\
& \multicolumn{8}{c}{$n=500$} \\
& \multicolumn{2}{c}{\texttt{FAMA}-CLT} 
& 
& \multicolumn{2}{c}{\texttt{FAMA}-$\tilde \Pi$} 
& 
& \multicolumn{2}{c}{\texttt{FABLE}} \\
&coverage ($\%$) & length & & coverage ($\%$) & length & & coverage ($\%$) & length \\
$\Lambda_m \Lambda_m^\top$    & 94.98 & 2.10 & &  94.98 & 2.10  & & 85.98 & 2.06 \\
$\Lambda_m \Lambda_{m'}^\top$ & 96.03 & 1.98 & & 96.02 & 1.98 & & 93.51 & 2.24 \\
& \multicolumn{8}{c}{$n=1000$} \\
& \multicolumn{2}{c}{\texttt{FAMA}-CLT} 
& 
& \multicolumn{2}{c}{\texttt{FAMA}-$\tilde \Pi$} 
& 
& \multicolumn{2}{c}{\texttt{FABLE}} \\
&coverage ($\%$) & length & & coverage ($\%$) & length & & coverage ($\%$) & length \\
$\Lambda_m \Lambda_m^\top$    &  94.69& 1.49 & & 94.69& 1.49 & & 84.94 & 1.49 \\
$\Lambda_m \Lambda_{m'}^\top$ & 95.59 & 1.41 & &  95.59 & 1.41  & & 94.26 & 1.59\\
\end{tabular}}

\end{table}

\section{Cancer Multiomics Application}\label{sec:application}

We analyze the data from \citet{linkedomics} available at \href{https://www.linkedomics.org}{https://www.linkedomics.org}. We consider four views collected from tumor samples, which include: 
(1) gene-level RNA sequencing (RNA-seq) data obtained from the Illumina GenomeAnalyzer platform, reported as normalized expression counts; (2) somatic copy number variation (SCNV) data at the gene level, quantified as log-ratios using the GISTIC2 pipeline to reflect copy number aberrations; (3) DNA methylation (DNA-met) data generated from the Illumina HM450K platform, summarized at the gene level as beta values; and (4) reverse phase protein array (RPPA) measurements providing gene-level, normalized protein expression profiles. The RPPA view contains 168 variables, while we retain the 4000 variables with the largest variance for each other view. We consider 170 samples for which data from each view are available, which are randomly divided into training and test sets of size 136 and 34 respectively. We fit each model on the training set after standard pre-processing, which we describe in the Supplementary Material. 

We compute the log-likelihood in the test set using point estimates for the overall covariance matrix of each method and test whether the out-of-sample log-likelihood of the \texttt{FAMA} estimate is larger than that of other methodologies by means of one-sided paired t tests. For \texttt{MOFA} and \texttt{ROTATE}, we fit the models with two different choices for the upper bound on the number of factors. More details are provided in the supplemental. The results are reported in Table \ref{tab:oos}. Overall \texttt{FAMA} produces excellent out-of-sample results that outperform competitors in almost all cases. Moreover, \texttt{FAMA} had a very competitive running time. Indeed, \texttt{FAMA} and \texttt{FABLE} only took $\sim 4$ seconds, while \texttt{MOFA} took 509 and 106 seconds, and \texttt{ROTATE} took 187 and 57 seconds.

\begin{table}[ht]
\caption{Comparison of out-of-sample log-likelihood of different methodologies. Out-of-sample log-likelihood of different methods and p-values of the paired one-sided t-tests for a greater log-likelihood for \texttt{FAMA} estimate. The first row refers to the concatenated dataset, the following four rows to each view separately, and the remaining rows to each possible pair of views. \label{tab:oos}}
\centering
\footnotesize
\resizebox{\textwidth}{!}{
\begin{tabular}{llcccccccc}
& & \texttt{FAMA} & \texttt{AJIVE}  & \texttt{FABLE} & \texttt{MOFA-1} & \texttt{MOFA-2} & \texttt{ROTATE-1} & \texttt{ROTATE-2} & \texttt{SVD} \\
\multirow{2}{*}{Overall}                 & log-lik & -106628.6 & -149324.9 & -140515.0 & -4371736.0 & -4408983.0& -152118.7 & -123055.9 & -109473.2 \\
                                         & p-value &  & $2 \times 10^{-3}$ & $8\times 10^{-8}$ & $0.04$ & $0.12$ & $8\times 10^{-8}$ & $8\times 10^{-8}$ & $0.43$ \\
\multirow{2}{*}{RNA-seq}                 & log-lik & -127946.0 & -161153.0 & -134057.5 & -128914.2 & -137424.2 & -145065.6 & -133610.3 & -135664.0 \\
                                         & p-value &  & $1 \times 10^{-13}$ & $2\times 10^{-7}$ & $0.03$ & $6\times 10^{-8}$ & $3\times 10^{-7}$ & $2\times 10^{-8}$ & $1\times 10^{-10}$ \\
\multirow{2}{*}{SCNV}                    & log-lik & 3435.1 & -31704.1 & -33278.9 & -4240487.0 & -4305088.0 & 3937.2 & -15005.9 & -31677.6 \\
                                         & p-value &  & $3 \times 10^{-5}$ & $7\times 10^{-17}$ & $0.04$ & $0.12$ & $0.52$ & $2\times 10^{-4}$ & $0.041$ \\
\multirow{2}{*}{DNA-met}                 & log-lik & 64677.7 & 45869.7 & 55418.0 & 60951.7 & 58374.3 & 65084.3 & 59163.8 & 58386.5 \\
                                         & p-value &  & $3 \times 10^{-9}$ & $3\times 10^{-7}$ & $0.10$ & $1 \times 10^{-4}$ & $0.60$ & $1\times 10^{-7}$ & $1.6\times 10^{-3}$ \\
\multirow{2}{*}{RPPA}                    & log-lik & -719.4 & 53.9 & -2088.1 & -1765.4 & -2337.2 & -1167.5 & -1888.7 & -362.6 \\
                                         & p-value &  & $1$ & $6\times 10^{-12}$ & $4\times 10^{-11}$ & $3\times 10^{-11}$ & $2\times 10^{-8}$ & $4\times 10^{-11}$ & $1.00$ \\
\multirow{2}{*}{\makecell{RNA-seq \&\\SCNV}}       & log-lik & -148123.6 & -194261.0 & -183485.9 & -4408290.0 & -4454298.0 & -189370.1 & -164391.1 & -167328.9 \\
                                                  & p-value &  & $9 \times 10^{-5}$ & $2\times 10^{-10}$ & $0.04$ & $0.13$ & $0.01$ & $1\times 10^{-4}$ & $0.14$ \\
\multirow{2}{*}{\makecell{RNA-seq \&\\DNA-met}}   & log-lik & -71350.3 & -116259.0 & -85061.7 & -74642.5 & -85183.3 & -92323.1 & -81535.5 & -77218.2 \\
                                                  & p-value &  & $2 \times 10^{-10}$ & $3\times 10^{-7}$ & $0.14$ & $2\times 10^{-9}$ & $9\times 10^{-6}$ & $5\times 10^{-8}$ & $0.011$ \\
\multirow{2}{*}{\makecell{RNA-seq \&\\RPPA}}   
                                                   & log-lik & -132550.0 & -161093.9 & -136851.5 & -131464.5 & -140158.0 & -148505.2 & -136481.4 & -135998.5 \\
                                                 & p-value &  & $9 \times 10^{-12}$ & $4\times 10^{-5}$ & $0.97$ & $2\times 10^{-6}$ & $8\times 10^{-7}$ & $1\times 10^{-5}$ & $8\times 10^{-5}$ \\
\multirow{2}{*}{\makecell{SCNV \&\\DNA-met}}  & log-lik & 48317.1 & 12731.4 & 10563.9 & -4211647.0 & -4258855.0 & 35098.5 & 27554.5 & 26732.2 \\
                                                  & p-value &  & $1 \times 10^{-3}$ & $4\times 10^{-11}$ & $0.04$ & $0.12$ & $0.20$ & $2\times 10^{-5}$ & $0.12$ \\
\multirow{2}{*}{\makecell{SCNV \&\\RPPA}}          & log-lik & -3615.8 & -31776.2 & -36637.4 & -4244299.0 & -4308319.0 & -1313.1 & -18429.7 & -32147.9 \\
                                                  & p-value &  & $1 \times 10^{-3}$ & $1\times 10^{-13}$ & $0.04$ & $0.12$ & $0.58$ & $5\times 10^{-4}$ & $0.070$ \\
\multirow{2}{*}{\makecell{DNA-met \&\\RPPA}}       & log-lik & 61122.7 & 45912.6 & 52556.5 & 58696.0 & 55765.0 & 62144.2 & 56284.4 & 58054.2 \\
                                                  & p-value &  & $6 \times 10^{-7}$ & $5\times 10^{-6}$ & $0.19$ & $2 \times 10^{-4}$ & $0.76$ & $2\times 10^{-5}$ & $0.056$ \\
\end{tabular}}
\end{table}

We also performed inference on the intra-view covariances, focusing on \texttt{FAMA} applied to the full 170 samples. For each view, we selected the 200 variables having absolute correlations greater than 0.5  with the largest number of other variables. For the RPPA, all the 168 variables were included. Figure \ref{fig:corr_plots} shows the correlation plot of the selected variables for each view.
\begin{figure}[htbp]
\centering
\begin{subfigure}{0.45\textwidth}
    \includegraphics[width=\linewidth]{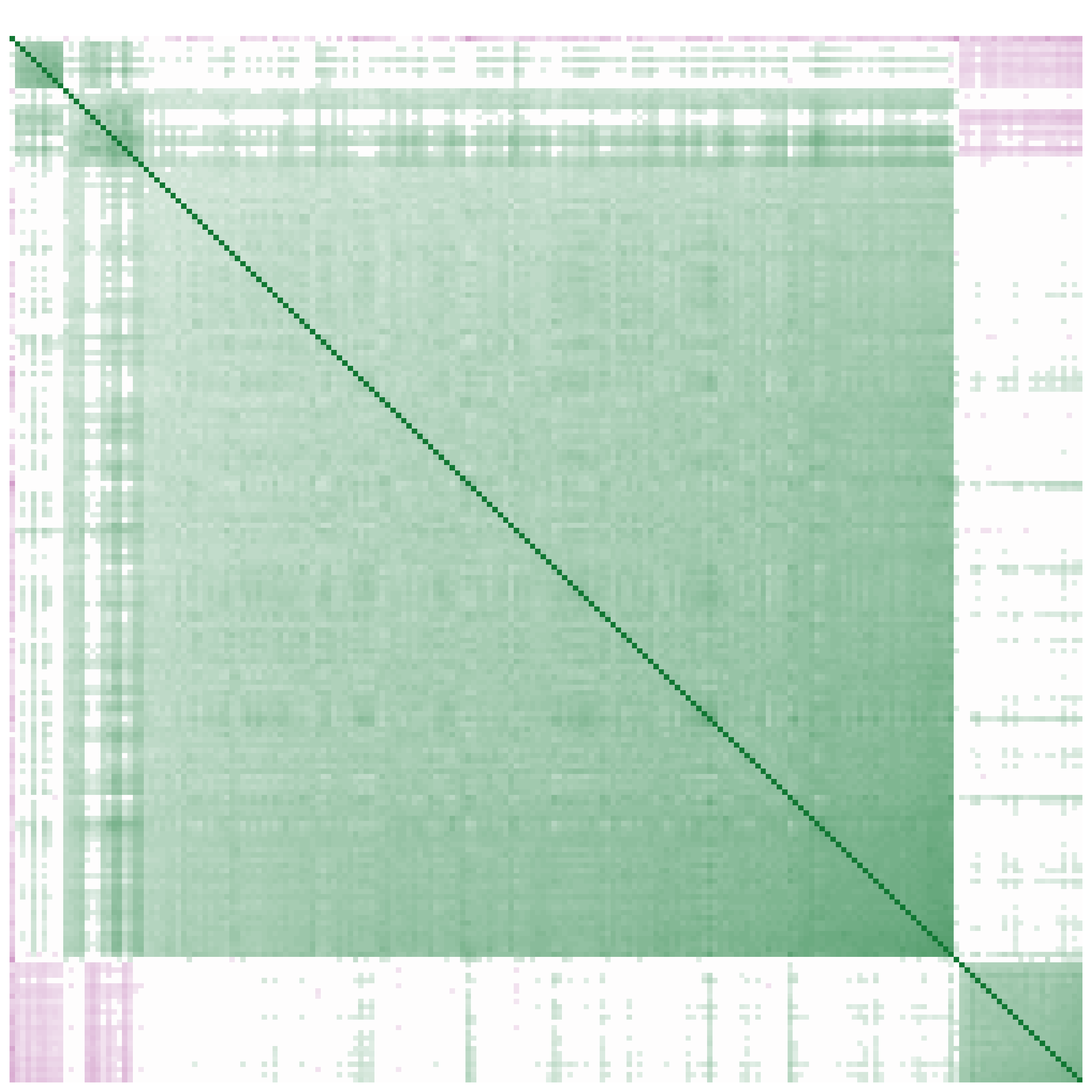}
    \caption{RNA sequencing (RNA-seq)}
\end{subfigure}
\hfill
\begin{subfigure}{0.45\textwidth}
    \includegraphics[width=\linewidth]{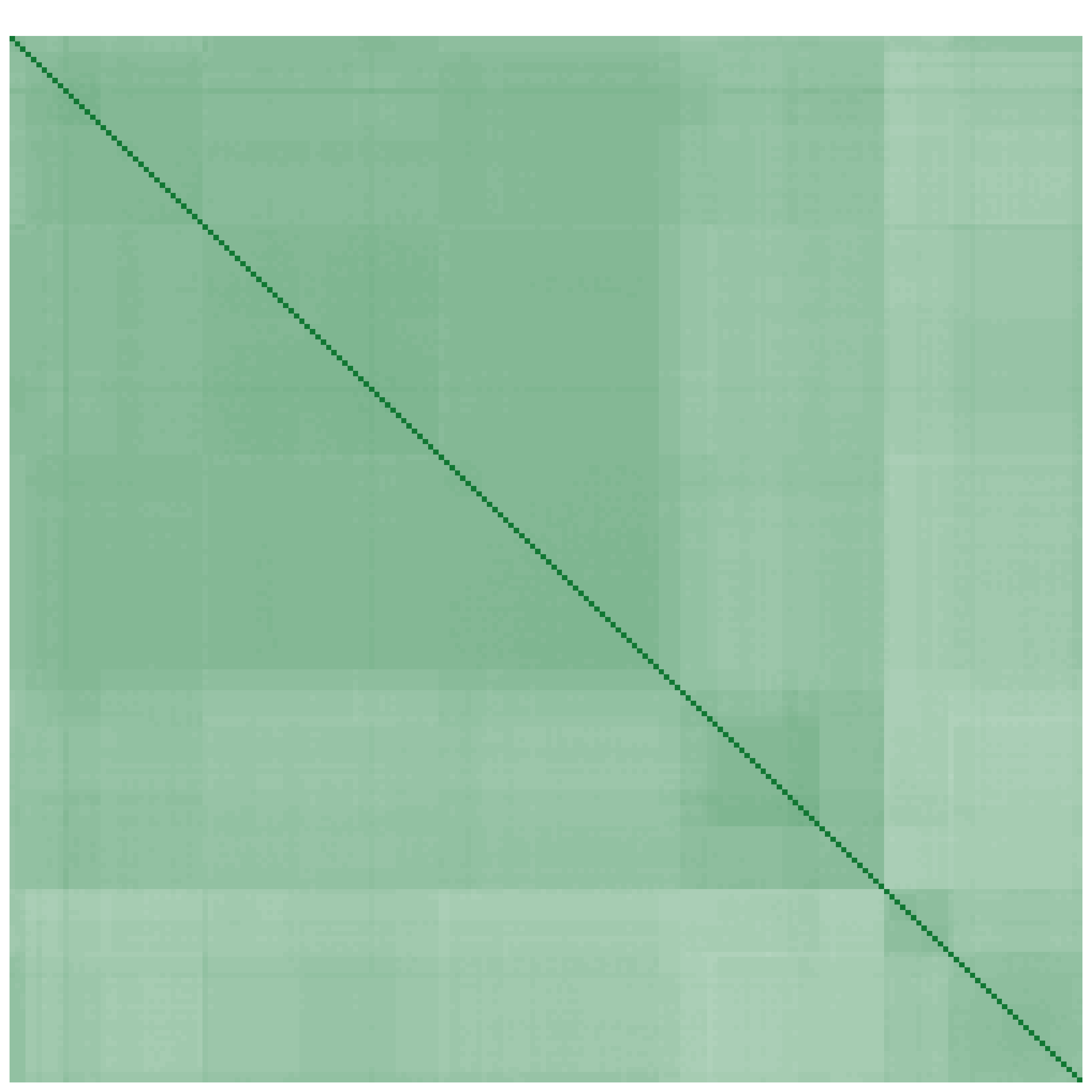}
    \caption{Somatic copy number variation (SCNV)}
\end{subfigure}
\vspace{0.5cm}
\begin{subfigure}{0.45\textwidth}
    \includegraphics[width=\linewidth]{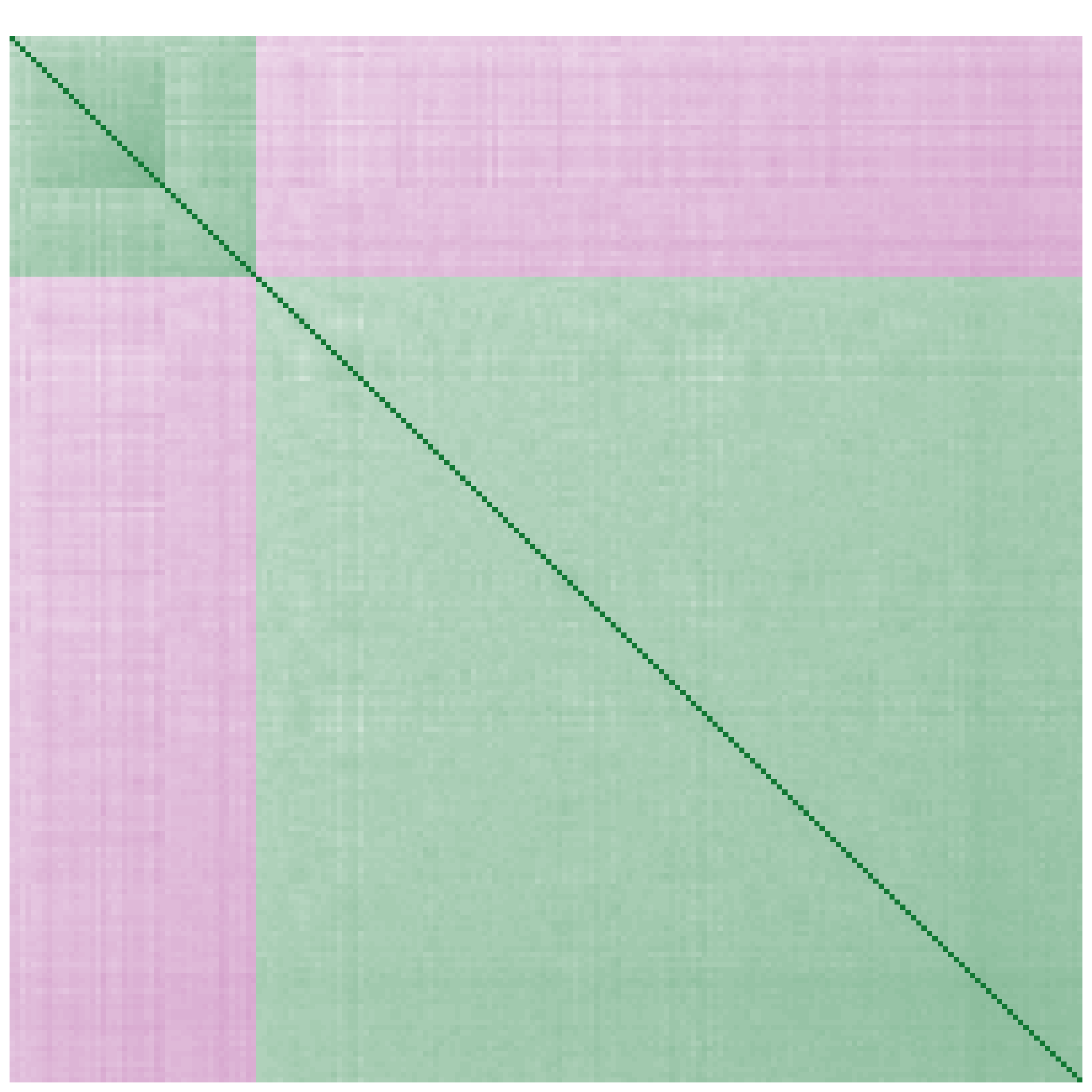}
    \caption{DNA methylation (DNA-met)}
\end{subfigure}
\hfill
\begin{subfigure}{0.45\textwidth}
    \includegraphics[width=\linewidth]{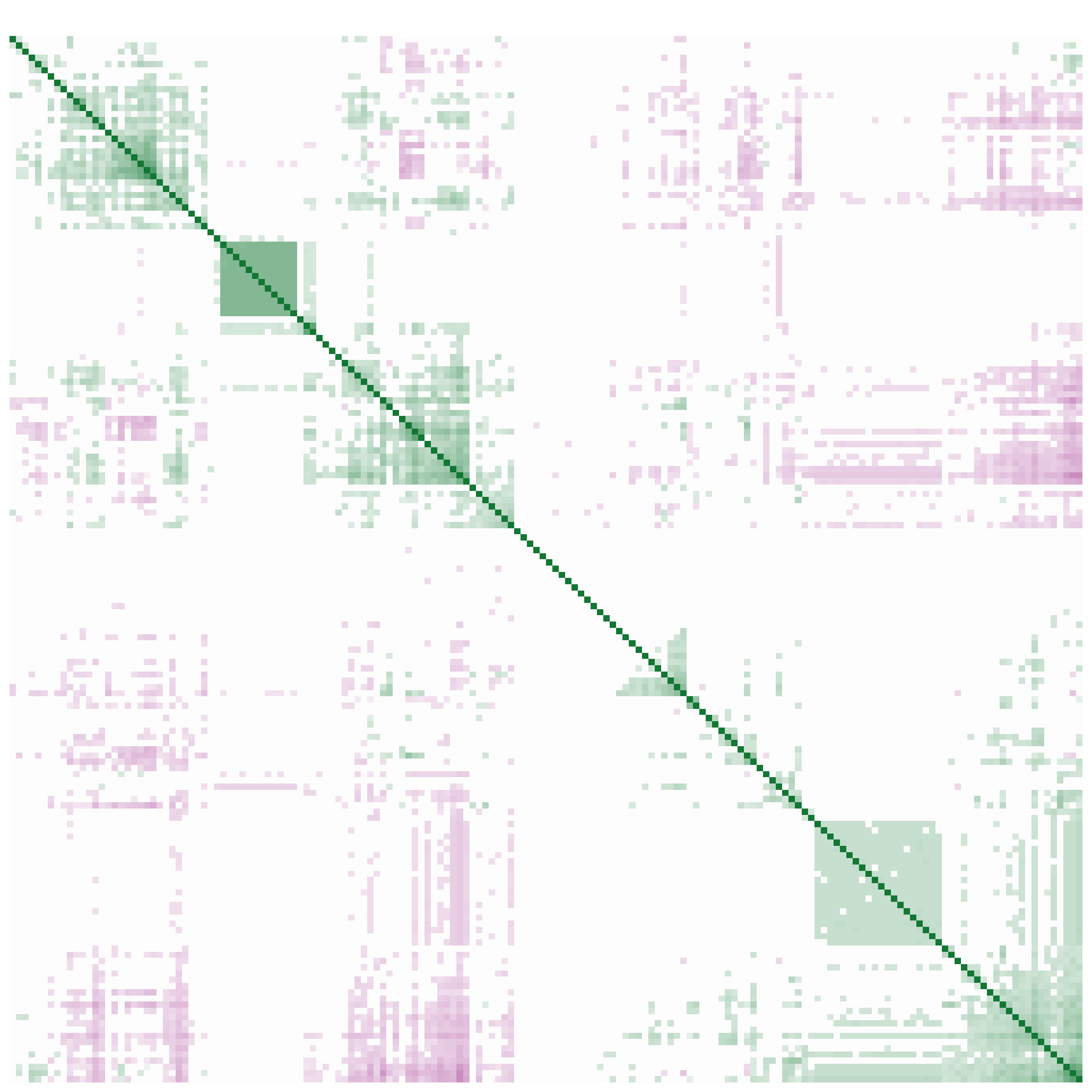}
    \caption{Reverse phase protein array (RPPA)}
\end{subfigure}
\caption{Within-view correlation plots inferred by \texttt{FAMA} for the cancer multiomics application. Positive (negative) correlations are showed in green (violet). Entries for which the $95\%$ posterior credible interval included 0 were set to 0.}
\label{fig:corr_plots}
\end{figure}
In the RNA-seq view, four clusters of genes emerge, with high positive intra-cluster correlations and mostly negative correlations across clusters. A prominent group consists of immune-associated genes (POU2AF1, CD38, ADAMDEC1, FAM30A, LCP1, CIITA), which are highly expressed in germinal center B cells or plasma cells \citep{immune_rna_2}. For example, POU2AF1 and FAM30A are part of a plasmablast gene module identified in studies of autoimmune disease and lymphoma \citep{immune_rna}. Another cluster includes genes that characterize epithelial or tumor-intrinsic programs (e.g. FOXA1, TTK, E2F8, FAP). For example, FOXA1 marks luminal breast/prostate tumors, which often exhibit reduced B-cell signatures \citep{nolan}. The SCNV view exhibits a more uniform behavior with a positive correlation among all features considered. There are 4 loci on chromosome 1q (TGFB2, LYPLAL1, C1orf143, RRP15) with a pairwise correlation greater than 0.75, consistent with a segment of chromosome 1q acquired together in a subset of tumors \citep{1q_amplification}. In the DNA-meth view, there is a bipartite structure with variables in the same (different) cluster(s) having strong positive (negative) correlation. One cluster contains neuronal genes (MAPT), DNA repair/apoptosis regulators (MLH3, HTT), and metabolic genes. The other includes FOXC2, a transcription factor linked to EMT (epithelial–mesenchymal transition) and development, COL9A3 (a collagen gene) and GPRC5B (a retinoic acid inducible gene) suggesting an EMT / mesenchymal and inflammation-associated methylation pattern. Such clustering could reflect a scenario in which a tumor subtype methylates a set of developmental genes and cell interaction genes, whereas some other subtype methylates a different set,  related to neural or differentiation pathways. Inverse methylation patterns are a known phenomenon in pan-cancer analyses \citep{inverse_met}. In the RPPA view, the correlation structure is more complex, but there are groups of proteins having high positive mutual correlation. One such group includes EMT markers (SNAI1, CDH2, ANXA1, GCN5, JAK2, GATA6) \citep{emt}. Another is related to the Receptor Tyrosine Kinase pathway (EGFR, MEK1, MAPK14 , YAP1, PDK1, PRKCA).

\begin{figure}
\centering
\begin{subfigure}{0.45\textwidth}
    \includegraphics[width=\linewidth]{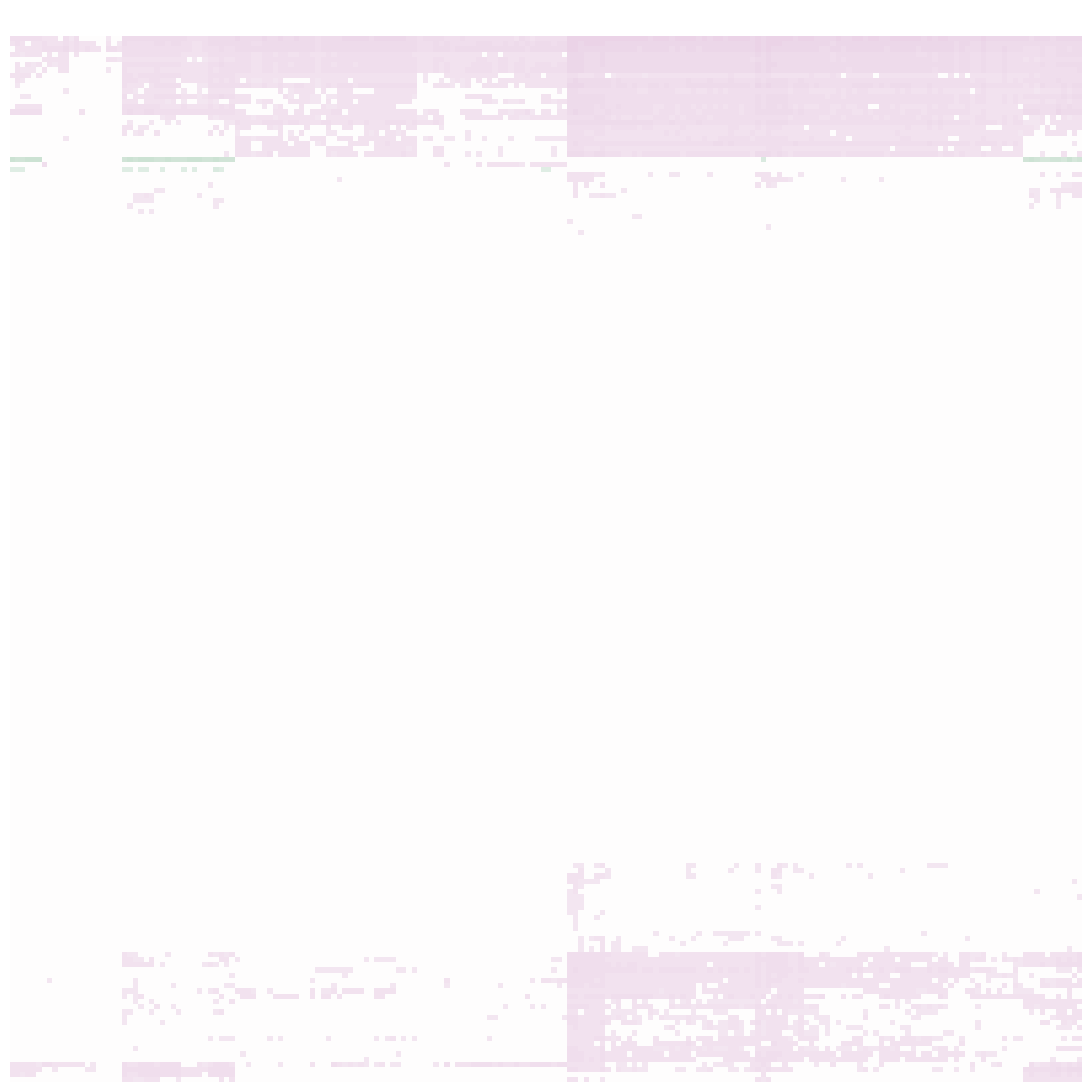}
    \caption{RNA-seq \& SCNV}
\end{subfigure}
\hfill
\begin{subfigure}{0.45\textwidth}
    \includegraphics[width=\linewidth]{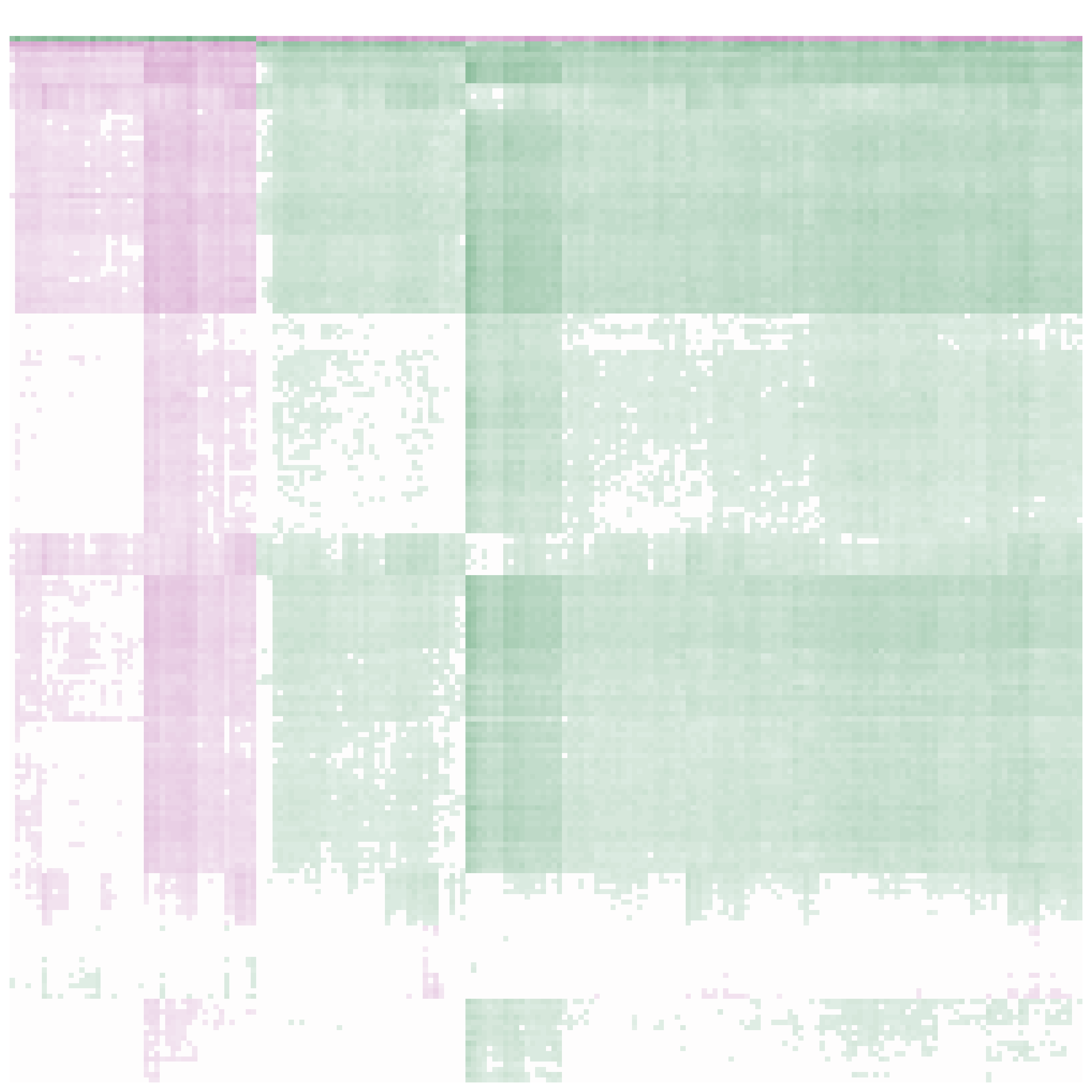}
    \caption{RNA-seq \& DNA-met}
\end{subfigure}
\vspace{0.5cm}
\begin{subfigure}{0.45\textwidth}
    \includegraphics[width=\linewidth]{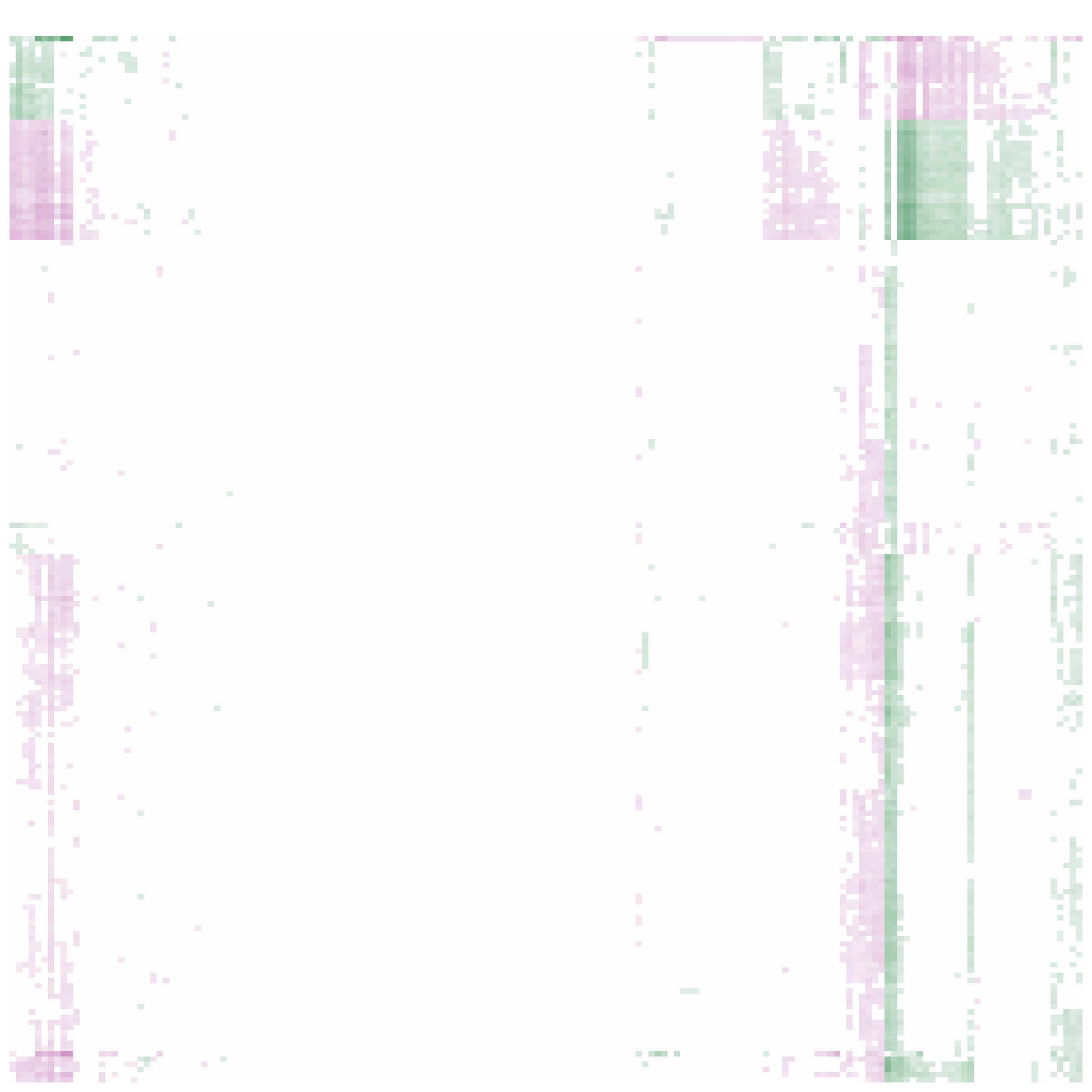}
    \caption{RNA-seq \& RPPA}
\end{subfigure}
\hfill
\begin{subfigure}{0.45\textwidth}
    \includegraphics[width=\linewidth]{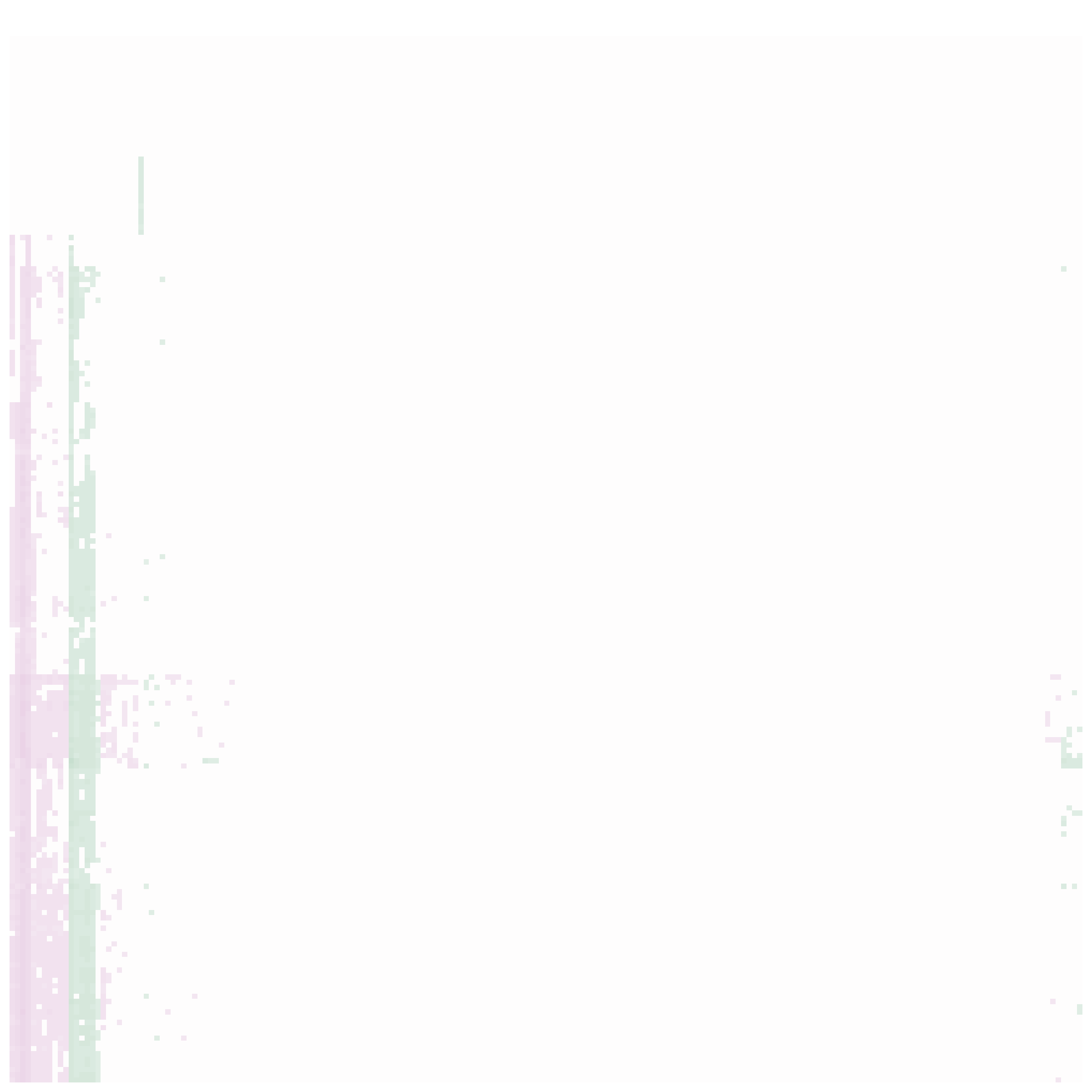}
    \caption{SCNV \& DNA-meth}
\end{subfigure}
\vspace{0.5cm}

\end{figure}


    \begin{figure}[htb]
        \ContinuedFloat

\begin{subfigure}{0.45\textwidth}
    \includegraphics[width=\linewidth]{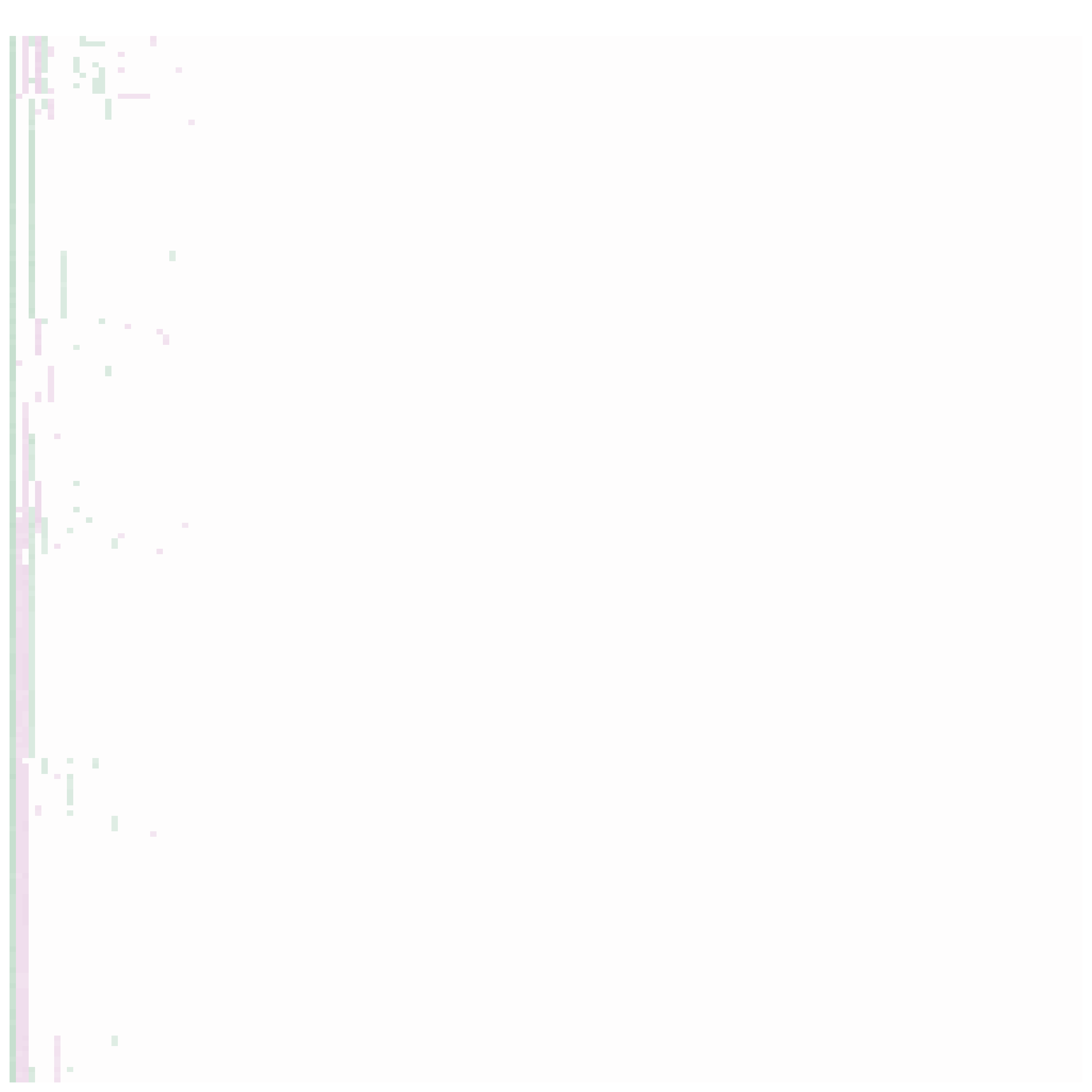}
    \caption{SCNV \& RPPA}
\end{subfigure}
\hfill
\begin{subfigure}{0.45\textwidth}
    \includegraphics[width=\linewidth]{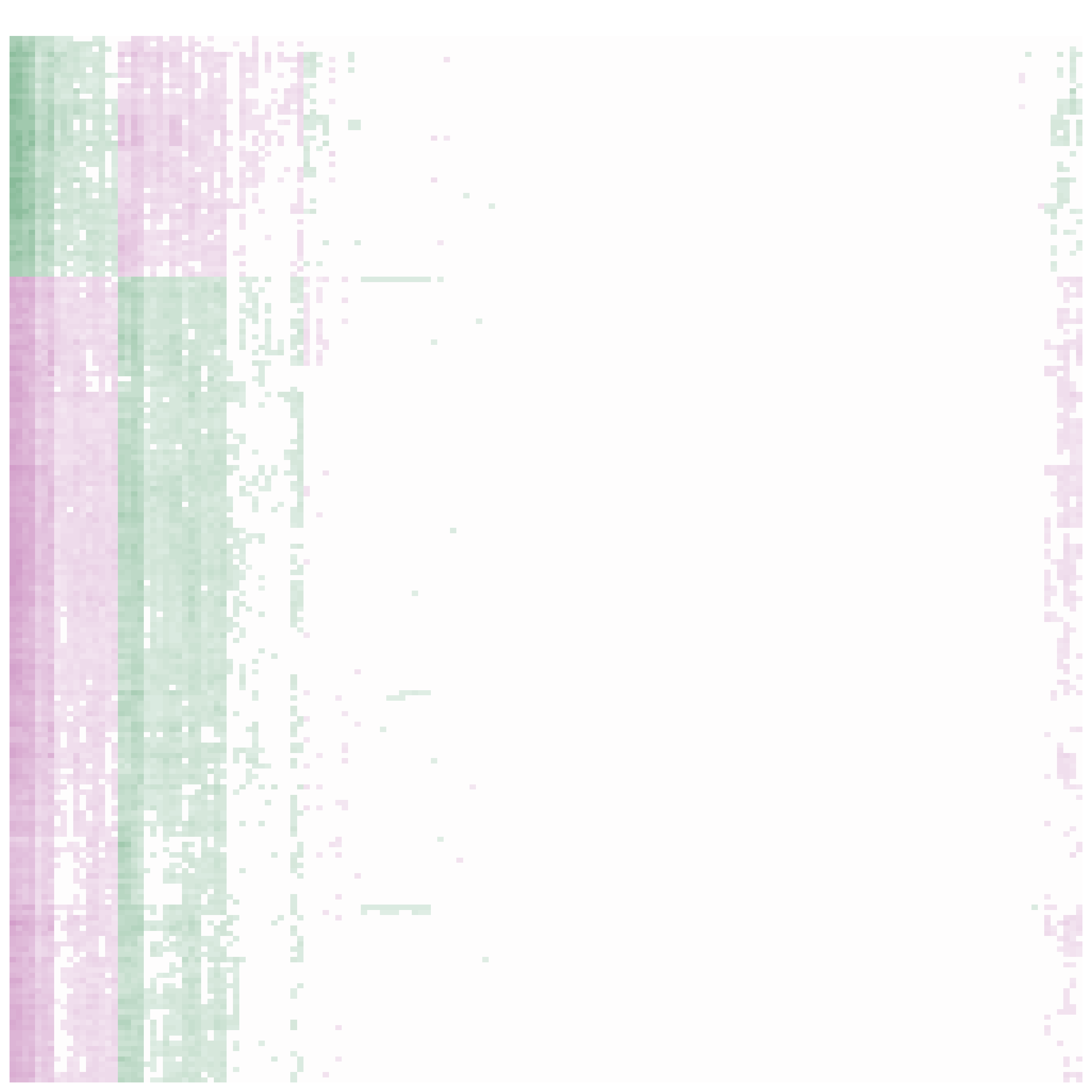}
    \caption{DNA-meth \& RPPA}
\end{subfigure}
\caption{Across-view correlation plots inferred by \texttt{FAMA} for the cancer multiomics application. Positive (negative) correlations are showed in green (violet). Entries for which the $95\%$ posterior credible interval included 0 were set to 0.}
\label{fig:corr_across_plots}
\end{figure}

Figure \ref{fig:corr_across_plots} shows across-view correlation plots. The RNA-seq \& SCNV correlation matrix contains mostly zero or slightly negative values, indicating that many gene copy-number changes do not translate into proportional RNA change \citep{Bhattacharya2020}. Among the negative correlations, the RNA value of IL18RAP (an interleukin-18 receptor accessory protein) is negatively correlated with copy number of RABGAP1L (a 1q25 gene) consistently with \citet{Hutchinson2017}.  The RNA-seq \& DNA-met correlation plot shows a clear split between positive and negative correlations, reflecting the dual role of DNA methylation in gene regulation \citep{Anastasiadi2018}. The RNA-seq \& RPPA matrix is relatively sparse with only a few moderate correlation values that indicate that RNA-seq levels are not strongly predictive of protein levels \citep{gry_rna_corr}. The SCNV \& RPPA correlation matrix is the sparsest with $\sim2$\% elements significantly different from zero consistent with \cite{akbani2014}. The SCNV \& DNA-met covariance is mostly sparse with a moderate majority of negative correlations among significant values. This tendency for copy number increases to associate with decreased methylation is consistent with \cite{prall06}. Finally, among the DNA-meth \& RPPA significant correlations, Estrogen Receptor (ER) $\alpha$ protein shows an inverse relationship with DNA methylation at the ESR1 gene locus. In ER-positive breast cancers, the ESR1 promoter is typically unmethylated and the ER protein is expressed, while in ER-negative tumors the ESR1 promoter is often hypermethylated, silencing the gene for the expression of ER protein \citep{ros20}. 

In summary, \texttt{FAMA}'s estimates have competitive out-of-sample performance and find highly interpretable biological signals.

\section{Conclusion}
We have presented a novel approach to inferring covariance structure in multi-view data with rigorous theoretical guarantees. \texttt{FAMA} estimates the latent factors that are active in at least one view using an eigendecomposition approach and infers the loadings through surrogate regression tasks. We have also proposed two ways to characterize uncertainty in inferring the covariance components, one based on a central limit theorem result, and one using the posterior distribution arising from the surrogate regression tasks. 

There are many interesting directions for future work. For example, in ecology, it has become standard to collect presence/absence or abundance data for large numbers of species at each sampling site. In this case, it is natural to treat the phylum that each species belongs to as a data view. Then, it becomes necessary to extend \texttt{FAMA} to account for binary or count data. One promising direction is to maintain the factor analytic form of \texttt{FAMA}, but within the linear predictor of a generalized linear latent variable model that would allow each variable to have a different measurement scale (continuous, count, binary, categorical). Along similar lines, it is useful to extend the methodology to other families of latent factor models, such as low-rank and non-negative matrix factorizations  \citep{zito2024compressive, msnmf}. In addition, it would be worthwhile to consider supervised variants, where latent factors are used to predict outcome variables of interest \citep{fin, fan1, fan2}. 

On the theoretical side, it would be interesting to to establish joint Bernstein–von Mises results for collections of covariance entries or more general functionals. Such extensions would provide a foundation for simultaneous inference and enable formal guarantees for multiplicity adjustment, including control of false discoveries when many covariance entries are considered jointly. 

Finally, given the promising results presented in this paper, applying \texttt{FAMA} to other high-dimensional multi-view datasets to obtain fast and reliable estimates of intra- and inter-omics dependence structure may lead to useful new applied insights.

\section*{Acknowledgement}
This research was partially supported by the National Institutes of Health (grant ID R01ES035625), by the Office of Naval Research (Grant N00014-24-1-2626), and by the European Research Council under the European Union’s Horizon 2020 research and innovation programme (grant agreement No 856506).

\bibliographystyle{agsm}
\bibliography{References}

\newpage
\appendix

\section*{Supplementary material for \say{Inference on covariance structure in high-dimensional multi-view data}}

\section{Proofs of theoretical results}

 \subsection{Preliminary results}
\begin{proposition}[Adapted from Proposition 3.5 of \citet{fable}]\label{prop:reovery_P_m}
    Let $U_m$ be the matrix of left singular vectors associated with the leading $k_m$ singular values of $Y_m$ and $U_{0m}$ be the matrix of left singular vectors of $F_{0m} = F_0 A_m$. Then, under Assumptions \ref{assumption:model}--\ref{assumption:hyperparameters}, as $n, p_m \to \infty$, with probability at least $1-o(1)$, we have
    \begin{equation*}
        ||P_m - P_{0m} || \lesssim \frac{1}{n} + \frac{1}{p_m},
    \end{equation*}
    where $P_m = U_m U_m^\top$ and $P_{0m} = U_{0m}U_{0m}^\top$.
\end{proposition}
\begin{proof}[of Proposition \ref{prop:reovery_P_m}]
    Since $U_{0m} = \tilde U_{0m} R$ for some orthogonal matrix $R$, where $\tilde U_{0m}$ is the matrix of left singular vectors of $F_{0m}\Lambda_{0m}^\top$, we have $U_{0m}U_{0m}^\top = \tilde U_{0m} \tilde U_{0m}^\top$. The result follows from Proposition 3.5 of \citet{fable}.
\end{proof}

\begin{proposition}\label{prop:recovery_P_0}
Let $\bar P$ be the orthogonal projection onto the subspace spanned by the singular vectors associated to the leading $k_0$  singular values of $\tilde P$, where $\tilde P$ is defined as in \eqref{eq:P_tilde}, and $U_0 \in  R^{n \times k_0}$ is the matrix of left singular vectors of $F_0$. Then, under Assumptions \ref{assumption:model}--\ref{assumption:hyperparameters}, we have
    \begin{equation}
         pr\left\{\left|\left|\bar P - U_0 U_0^\top \right|\right| > G_2 \left(\frac{1}{n} + \frac{1}{p_{\min}}\right)\right\} \to 0 \quad (n, p_1, \dots, p_M \uparrow \infty),
    \end{equation}
where $G_2 < \infty$, and $p_{\min} = \min_{m=1, \dots, M}p_m$. 
\end{proposition}
\begin{proof}[of Proposition \ref{prop:recovery_P_0}]
Note that $\tilde P = \frac{1}{M} \sum_{m=1}^M P_m = \frac{1}{M} \sum_{m=1}^M (U_{0m} U_{0m}^\top + \Delta_m)$, where $\Delta_m = P_m - P_{0m}$. By Proposition \ref{prop:reovery_P_m}, we have $ ||\frac{1}{M} \sum_{m=1}^M \Delta_m || \lesssim \frac{1}{M} \sum_{m=1}^M \Big(\frac{1}{n} + \frac{1}{p_m}\Big) \asymp \frac{1}{n} + \frac{1}{p_{\min}}$, with probability at least $1-o(1)$. 
Moreover, we have $U_{0m} = U_0 R_m$ for $R_m = U_0^\top U_m \in \mathbb R^{k_0 \times k_0}$. Note that $R_m =  D_0^{-1} V_0 F_0^\top F_0 A_m V_{0m} D_{0m}^{-1}$, where $V_0$ and $V_{0m}$ are the matrices of right singular vectors of $F_0$ and $F_{0m }$, and  $D_0$ and $D_{0m}$ are diagonal matrices with the corresponding singular values on the diagonal.
Hence, $\frac{1}{M} \sum_{m=1}^M P_{0m} = U_0 D_0 V_0^\top \big(\frac{1}{M} \sum_{m=1}^M A_m V_{0m} D_{0m}^{-2} V_{0m}^\top A_m^\top \big)V_0 D_0 U_0^\top$. In particular, 
\begin{equation*}
    \begin{aligned}
        s_{k_0}\bigg(\frac{1}{M} \sum_{m=1}^M P_{0m}\bigg) & \geq s_{k_0}(U_0 D_0 V_0^\top )^2 s_{k_0} \bigg(\frac{1}{M} \sum_{m=1}^M A_m V_{0m} D_{0m}^{-2} V_{0m}^\top A_m^\top \bigg)\\
        & \geq \frac{1}{M}  s_{k_0}(D_0)^2  s_{k_0}\bigg(\sum_{m=1}^M A_m V_{0m} D_{0m}^{-2} V_{0m}^\top A_m^\top\bigg).
    \end{aligned}
\end{equation*}
We lower bound the second term. For all $m=1, \dots, M$, $A_m V_{0m} D_{0m}^{-2} V_{0m}^\top A_m^\top \succeq d_{0m 1}^{-2}A_m V_{0m}  V_{0m}^\top A_m^\top = d_{0m 1}^{-2} A_m  A_m^\top $, where $d_{0m k_m}$ is the $k_{m}$-th entry on the diagonal of $D_{0m}$. We let $d^2 = \max_{m=1, \dots, M} d_{0m 1}$. Then, we have
$ s_{k_0}(\sum_{m=1}^M A_m V_{0m} D_{0m}^{-2} V_{0m}^\top A_m^\top) \geq d^{-2} s_{k_0}(\sum_{m=1}^M A_m A_m^\top) = d^{2}$, since $\sum_{m=1}^M A_m A_m^\top$ is a diagonal matrix where the $l,l$-th entry is the number of views in which the $l$-th column of $F_0$ is active and, therefore is larger than or equal to $1$.
Combining all of the above, we obtain
\begin{equation*}
    \begin{aligned}
        s_{k_0}\bigg(\frac{1}{M} \sum_{m=1}^M P_{0m}\bigg) & \geq  \frac{1}{M}  s_{k_0}(D_0)^2  d^{-2} \asymp \frac{1}{M},
    \end{aligned}
\end{equation*}
with probability at least $1- o(1)$, since by Corollary 5.35 in \citet{vershynin_12}, we have $s_{k_0}(D_0) = \sqrt{n} - \mathcal{O}(k_0)$ and $d = \sqrt{n}+ \mathcal{O}(\sum_{m}k_m)$ with probability at least $1- o(1)$.
Therefore, denoting with 
$E = \frac{1}{M} \sum_{m=1}^M \Delta_m$, 
we have $||E || < 1/M$ with probability at least $1-o(1)$, for $n$ and $p_{\min}$ sufficiently large. Thus, we can apply Theorem 2.1 from \citet{gratton_16}, which characterizes the matrix left singular vectors associated to the $k_0$ leading singular values of $\tilde P$ up to an orthogonal transformation as 
\begin{equation}\label{eq:left_sv_tilde_P}
    \bar U = (U_0 -  U' S)(I_{k_0} +  S^\top  S)^{-1},
\end{equation}
where $ U' \in \mathbb R^{n \times n - k_0} $ is a matrix spanning the orthogonal complement of $\text{span}(U_0)$ such that $U^{' \top} U' = I_{n-k_0}$, and $ S \in  R^{n - k_0 \times k_0}$ such that $||S|| \asymp || E || \lesssim \frac{1}{n} + \frac{1}{p_{\min}}$ with probability at least $1- o(1)$. From \eqref{eq:left_sv_tilde_P}, we can get an expression for $\bar P = \bar U \bar U^\top$ and applying the Woodbury identity, we obtain
\begin{equation*}
    ||\bar P - U_0 U_0^\top|| \lesssim \frac{1}{n} + \frac{1}{p_{\min}},
\end{equation*}
 with probability at least $1- o(1)$, yielding the desired result.
\end{proof}

\subsection{Proof of the main results}

\begin{proof}[of Theorem \ref{thm:recovery_factors_overall}]
    Let $U_0 D_0 V_0^\top $ be the singular value decomposition of $F_0$. By Proposition \ref{prop:recovery_P_0}, we have $||\bar U \bar U^\top  - U_0 U_0^\top ||  \lesssim \frac{1}{n} + \frac{1}{p_{\min}}$, with probability at least $1 - o(1)$.
    Then, by Davis-Kahan theorem \citep{davis_kahan} we have
    $\min_{R: R^\top R = I_{k_0}}||U  - U_0R || = ||U  - U_0 \hat R || \lesssim \frac{1}{n} + \frac{1}{p_{\min}}$, with probability at least $1 - o(1)$, where $\hat R$ is the orthogonal matrix achieving the minimum of the quantity on the left hand side. Recalling that $\hat F = \sqrt{n} U$ and letting $\tilde R = V_0 \hat R$, we have \begin{equation*}
        \begin{aligned}
            ||\hat F - F_0 \tilde R  || &= ||\sqrt{n}U - U_0 D_0 \hat R || \leq  ||\sqrt{n} (U -U_0\hat R ) || + ||\sqrt{n} U_0\hat R - U_0 D_0 \hat R  ||\\
            &\leq \sqrt{n} \big(\frac{1}{n} + \frac{1}{p_{\min}}\big) + \max_{1 \leq l \leq k_0} |\sqrt{n} - d_{0l}|,
        \end{aligned}
    \end{equation*}where $d_{0l}$ is the $l$-th largest singular value of $D_0$. Moreover, by Corollary 5.35 of \citet{vershynin_12}, we have $|d_{0l} - \sqrt{n}| \lesssim \sqrt{k_0}$ with probability at least $1-o(1)$. 
\end{proof}

\begin{proof}[of Theorem \ref{thm:posterior_contraction}]
We start by establishing consistency of point estimates. Recall that $\hat \Lambda_m = Y_m^\top U \sqrt{n} / (n + \tau_m^{-2})$. Let $m, m' \in \{1, \dots, M\}$, not necessarily distinct, and $g_{mm'} = \frac{n}{(n + \tau_m^{-2})(n + \tau_{m'}^{-2})}$  then
\begin{equation*}
    \begin{aligned}
        \hat \Lambda_m \hat \Lambda_{m'}^\top &= g_{mm'}  Y_m^\top U U^\top Y_{m'}=  g_{mm'}  Y_m^\top U_0 U_0^\top Y_{m'}  + g_{mm'}  Y_m^\top \big(U U^\top -  U_0 U_0^\top\big) Y_{m'}
    \end{aligned}
\end{equation*}
and 
\begin{equation*}
    \begin{aligned}
         \frac{n}{(n + \tau_m^{-2})^2} Y_m U_0 U_0^\top Y_{m'}  = \frac{n}{(n + \tau_m^{-2})^2} & \big(F_{0} A_m \Lambda_{0m}^\top + E_m\big)^\top U_0 U_0^\top\big(F_{0} A_{m'} \Lambda_{0m'}^\top + E_{m'}\big)\\
         = \frac{n}{(n + \tau_m^{-2})^2} &\big( \Lambda_{0m} A_m^\top F_{0}^\top  F_{0}A_{m'}  \Lambda_{0m'}^\top + E_m^\top F_{0} A_{m'}\Lambda_{0m'}^\top \\
         & \quad +  \Lambda_{0m}A_m^\top F_{0}^\top E_{m'} + E_m^\top U_0 U_0^\top E_{m'}\big).
    \end{aligned}
\end{equation*}
Moreover, since 
$$A_m^\top F_{0}^\top  F_{0}A_{m'} = \begin{cases}
    F_{0m}^\top F_{0m}, \quad &\text{if } m=m',\\
    A_{m} B_{mm'} F_{0mm'}^\top F_{0mm'}  B_{mm'}^\top A_{m'}, \quad &\text{otherwise},\\
\end{cases}$$
where $B_{mm'} \in \mathbb R^{k_0 \times k_{mm'}^c}$ maps the columns of matrix $F_{0mm'}^c$ to the corresponding columns of $F_0$. Thus, for $m=m'$, we have
\begin{equation*}
    \begin{aligned}
       g_{mm}   \Lambda_{0m} A_m^\top F_{0}^\top  F_{0}A_{m'} \Lambda_{0m'}^\top & = g_{mm}  \big[\Lambda_{0m} n I_{k_m} \Lambda_{0m}^\top + \Lambda_{0m}\big(F_{0m}^\top F_{0m} - n I_{k_m}\big) \Lambda_{0m}^\top\big],
    \end{aligned}
\end{equation*}
and, for $m\neq m'$,
\begin{equation*}
    \begin{aligned}
      g_{mm'}  \Lambda_{0m} A_m^\top F_{0}^\top  F_{0}A_{m'} \Lambda_{0m'}^\top  = g_{mm'} &\big[\Lambda_{0m}A_m^\top n I_{k_{mm'}^c} A_{m'} \Lambda_{0m'}^\top \\
      & \quad + \Lambda_{0m}A_{m}^\top B_{mm'} \big(F_{0mm'}^{c\top} F_{0mm'}^c - n I_{k_{mm}^c}\big) \Pi_{mm'}^\top B_{mm'}^\top \Lambda_{0m'}^\top\big].
    \end{aligned}
\end{equation*}
since $B_{mm'}B_{mm'}^{\top}$ is a diagonal matrix with the $l$-th diagonal element equal to $1$ if and only if the $l$-th factor in $F_0$ is active in views $m$ and $m'$, which implies $A_m^\top B_{mm'}B_{mm'}^{\top} A_{m'} = A_m^\top  A_{m'} $. 
Then, we have 
\begin{equation*}
    \begin{aligned}
     ||  g_{mm'} n \Lambda_{0m}  \Lambda_{0m'}^\top - \Lambda_{0m}  \Lambda_{0m'}^\top|| \lesssim \frac{\sqrt{p_m p_{m'}}}{n},
    \end{aligned}
\end{equation*} 
\begin{equation*}
    \begin{aligned}
       g_{mm'}   || \Lambda_{0m}\big(F_{0m}^\top F_{0m} - n I_{k_m}\big) \Lambda_{0m}^\top || & \lesssim \frac{1}{n}||\Lambda_{0m}||^2 ||F_{0m}^\top F_{0m} - n I_{k_m}  || \asymp \frac{p_m}{\sqrt{n}}, 
    \end{aligned}
\end{equation*}
and 
\begin{equation*}
    \begin{aligned}
       g_{mm'}  &|| \Lambda_{0m} n I_{k_{mm'}^c} \Lambda_{0m'}^\top + \Lambda_{0m}A_{m}^\top B_{mm'}\big(F_{0mm'}^{c\top} F_{0mm'}^c - n I_{k_{mm}^c}\big) B_{mm'}^\top A_{m'} \Lambda_{0m'}^\top  || \\
&\lesssim \frac{1}{n}||\Lambda_{0m}||||\Lambda_{0m'}|| ||F_{0m}^\top F_{0m} - n I_{k_m}  || ||B_{mm'}||^2 ||A_m|| ||A_{m'}|| \asymp \frac{\sqrt{p_m p_{m'}}}{\sqrt{n}}, 
    \end{aligned}
\end{equation*}
with probability at least $1 - o(1)$, since $||\Lambda_{0m}|| \asymp \sqrt{p_m}$ and  $||\Lambda_{0m'}|| \asymp \sqrt{p_{m'}}$ by Assumption \ref{assumption:Lambda} and
$||F_{0m}^\top F_{0m} - n I_{k_m}  || \asymp ||F_{0mm'}^\top F_{0mm'} - n I_{k_{mm'}^c}  ||  \asymp \sqrt{n} + \sqrt{k_{mm'}^c}$ by Corollary 5.35 of \citet{vershynin_12} with probability at least $1 - o(1)$,  $ g_{mm'} = \frac{1}{n} + \mathcal{O}(1/n^2)$ and $||B_{mm'}|| \asymp ||A_m|| \asymp ||A_{m'}|| \asymp 1$. 

In addition, with probability at least $1 - o(1)$, we have 
\begin{equation*}
    \begin{aligned}
        ||E_m^\top  F_{0m'} \Lambda_{0m'}^\top  || &\lesssim (\sqrt{n} + \sqrt{p_m}) \sqrt{n p_{m'}},\\
        ||E_{m'}^\top  F_{0m} \Lambda_{0m}^\top  || &\lesssim (\sqrt{n} + \sqrt{p_{m'}}) \sqrt{n p_{m}}
    \end{aligned}
\end{equation*}
and 
$$
||E_m^\top U_0 U_0^\top E_{m'}|| \lesssim (\sqrt{n} + \sqrt{p_m})(\sqrt{n} + \sqrt{p_{m'}})
$$
since $||E_m|| \lesssim \sqrt{n} + \sqrt{p_m}$, $||E_{m'}|| \lesssim \sqrt{n} + \sqrt{p_{m'}}$ $|| F_{0m} || \asymp || F_{0m'} || \lesssim \sqrt{n}$, with probability at least $1 -o(1)$, by Lemma \ref{lemma:Y_m} and Corollary 5.35 of \citet{vershynin_12}, respectively, and $||\Lambda_{0m} || \lesssim \sqrt{p_m}$,   $||\Lambda_{0m'}|| \asymp \sqrt{p_{m'}}$ Assumption \ref{assumption:Lambda}. Finally, with probability at least $1 - o(1)$,
\begin{equation*}
    \begin{aligned}
       g_{mm'}||Y_m \big(U U^\top -  U_0 U_0^\top\big) Y_{m'}^\top || &\lesssim ||Y_m || ||Y_{m'} || || U U^\top -  U_0 U_0^\top||\\
       & \lesssim \frac{1}{n} n \sqrt{p_m p_{m'}} \big(\frac{1}{n} + \frac{1}{p_{\min}}\big)
    \end{aligned}
\end{equation*}
 since, with probability at least $1 - o(1)$, $||Y_m|| \asymp \sqrt{n p_m}$, $||Y_{m'}||  \asymp \sqrt{n p_{m'}}$ by Lemma \ref{lemma:Y_m}, and $|| U U^\top -  U_0 U_0^\top|| \lesssim \frac{1}{n} + \frac{1}{p_{\min}}$ by Proposition \ref{prop:recovery_P_0}. 
Combining all of the above with $||\Lambda_{0m}  \Lambda_{0m}^\top|| \asymp p_m$ which is implied by Assumption $\ref{assumption:Lambda}$, for $m=m'$, we get that with probability at least $1 - o(1)$, 
\begin{equation*}
    \begin{aligned}
       \frac{||\hat \Lambda_m\hat \Lambda_m^\top  - \Lambda_{0m}  \Lambda_{0m}^\top||}{||\Lambda_{0m}  \Lambda_{0m}^\top||} &\lesssim \frac{1}{\sqrt{n}} + \frac{1}{\sqrt{p_m}} + \frac{1}{p_{\min}}.
    \end{aligned}
\end{equation*}
For $m \neq m'$, under the condition $||\Lambda_{0m}A_m^\top A_{m'}  \Lambda_{0m}^\top|| \asymp \sqrt{p_m p_{m'}}$, with probability at least $1 - o(1)$, we have 
\begin{equation*}
    \begin{aligned}
       \frac{||\hat \Lambda_m\hat \Lambda_{m'}^\top  - \Lambda_{0m}A_m^\top A_{m'}  \Lambda_{0m'}^\top||}{||\Lambda_{0m} A_m^\top A_{m'} \Lambda_{0m'}^\top||} &\lesssim \frac{1}{\sqrt{n}} + \frac{1}{\sqrt{p_m}} + \frac{1}{\sqrt{p_{m'}}} + \frac{1}{p_{\min}}.
    \end{aligned}
\end{equation*}

Next, we show posterior contraction. Recall that a sample for $\tilde \Lambda_l$ from $\tilde \Pi$ admits the representation
\begin{equation}
 \tilde \Lambda_l = \hat \Lambda_l + \tilde E_{\Lambda_l}, \quad  \tilde E_{\Lambda_l} = [ \tilde e_{\lambda_{l1}} ~ \cdots ~\tilde e_{\lambda_{lp_l}}]^\top, \quad  \tilde e_{\lambda_{l j}} \sim N_{k_0}\bigg(0, \frac{\rho_l ^2 \tilde \sigma_{lj}^2}{n + \tau_l^{-2}} I_{k_0}\bigg), \quad \tilde \sigma_{lj} \sim IG ( \nu_n/2, \nu_n \delta_{lj}^2/2 )
\end{equation}
with $j=1, \dots, p_l$ and $l=m, m'$. 
In particular, 
\begin{equation*}
\begin{aligned}
    ||\tilde E_{\Lambda_l}|| & \lesssim \rho_l \frac{\sqrt{p_l}}{\sqrt{n}} \max_{j=1, \dots p_l} \tilde \sigma_{lj} \lesssim  \frac{\sqrt{p_l}}{\sqrt{n}},
\end{aligned}
\end{equation*}
with $\tilde \Pi$-probability at least $1-o(1)$, since $ \max_{j}\tilde \sigma_{mj}^2  \lesssim 1$ with probability at least $1-o(1)$, by Lemma \ref{lemma:convergence_tilde_sigma}, which implies 
\begin{equation*}
\begin{aligned}
    ||\tilde \Lambda_m \tilde \Lambda_{m'}^\top - \hat \Lambda_m \hat \Lambda_{m'}^\top|| &\lesssim ||\tilde E_{\Lambda_m}||||\tilde E_{\Lambda_{m'}}|| + ||\hat \Lambda_m || |||\tilde E_{\Lambda_{m'}}|| + ||\hat \Lambda_{m'} || |||\tilde E_{\Lambda_m}||\\
  & \lesssim \frac{\sqrt{p_m p_{m'}}}{\sqrt{n}},
\end{aligned}
\end{equation*}
with probability at least $1 - o(1)$, since $||\hat \Lambda_m|| \lesssim \sqrt{p_m}$ and $||\hat \Lambda_{m'}|| \lesssim \sqrt{p_{m'}}$ with probability at least $1-o(1)$, by Lemma \ref{lemma:Lambda_m_hat}, yielding the desired result.
\end{proof}

\begin{proof}[of Theorem \ref{thm:clt}]
Recall $\hat \lambda_{lj} = \frac{\sqrt{n}}{n + \tau_l^{-2}}U^\top y_l^{(j)}$, where $y_l^{(j)}$ denotes the $j$-th column of $Y_l$.
Thus, 
\begin{equation*}
    \begin{aligned}
       \hat \lambda_{mj}^\top \hat \lambda_{m'j'} &= g_{mm'}  y_m^{(j)\top } U U^\top y_{m'}^{(j') }\\
      & =g_{mm'}  y_m^{(j)\top } U_0 U_0^\top y_{m'}^{(j') } + g_{mm'}\big[y_m^{(j)\top }\big(U U^\top - U_0 U_0^\top\big) y_{m'}^{(j') }\big].
    \end{aligned}
\end{equation*}
where $g_{mm'} = \frac{n}{(n + \tau_m^{-2})(n + \tau_{m'}^{-2})}$ and we note $g_{mm'} = \frac{1}{n} + h_{mm'}$, with $h_{mm'} \asymp O(1/n^2)$. We decompose the terms:
\begin{equation*}
    \begin{aligned}
        y_m^{(j)\top } U_0 U_0^\top y_{m'}^{(j')}& = 
        & = \lambda_{0mj}^\top A_m^\top F_0^\top F_0 A_{m'} \lambda_{0m'j'} + \lambda_{0mj}^\top A_m^\top F_0^\top e_{m'}^{(j')}  + \lambda_{0m'j'}^\top A_{m'}^\top F_0^\top e_{m}^{(j)} + e_{m}^{(j)\top} U_0 U_0^\top e_{m'}^{(j')} 
     \end{aligned}
\end{equation*}
Since $U_0^\top e_{l}^{(h)} \sim N_{k_0}(0, \sigma_{0lh}^2 I_{k_0})$, we have $|| U_0^\top e_{l}^{(h)} || \lesssim 1$ with probability at least $1-o(1)$ by Corollary 5.35 of \citet{vershynin_12} since $\sigma_{0lh} = \mathcal O(1)$ by Assumption \ref{assumption:sigma}. Hence, 
$$
\sqrt{n} \frac{1}{n} |e_{m}^{(j)\top} U_0 U_0^\top e_{m'}^{(j')} | \leq || U_0^\top e_{m}^{(j)}  || || U_0^\top e_{m'}^{(j')} || \lesssim \frac{1}{\sqrt{n}}  
$$
Moreover, 
\begin{equation*}
    \begin{aligned}
       \sqrt{n} h_{mm'} | y_m^{(j)\top } U_0 U_0^\top y_{m'}^{(j')} | \lesssim  \frac{1}{n^{3/2}} ||U_0^\top y_m^{(j)}|| || U_0^\top y_{m'}^{(j')} ||  \asymp \frac{1}{\sqrt{n} }
    \end{aligned}
\end{equation*}
with probability at least $1-o(1)$, since $||U_0^\top y_m^{(j)}||  \leq ||y_m^{(j)}||$, $  ||U_0^\top y_{m'}^{(j')} ||\leq ||y_{m'}^{(j')}||$ and $ ||y_m^{(j)}|| \asymp ||y_{m'}^{(j')}|| \lesssim \sqrt{n}$ with probability at least $1-o(1)$ by Lemma \ref{lemma:y_j}, and 
\begin{equation*}
    \begin{aligned}
     \sqrt{n} g_{mm'}  |y_m^{(j)\top }\big(U U^\top - U_0 U_0^\top\big) y_{m'}^{(j')}|  &\lesssim \sqrt{n} \frac{1}{n}|| y_m^{(j) }||  ||y_{m'}^{(j')} || || U U^\top - U_0 U_0^\top || \\
     & \lesssim  \sqrt{n} \big(\frac{1}{n} + \frac{1}{p_{\min}}\big) \asymp \frac{1}{\sqrt{n}} + \frac{\sqrt{n}}{p_{\min}},
    \end{aligned}
\end{equation*}
with probability least $1-o(1)$, 
which vanishes under the condition $\sqrt{n} / p_{\min} = o(1)$,  since,  with probability at least $1-o(1)$, $||y_m^{(j)}|| \asymp || y_{m'}^{(j')} || \asymp \sqrt{n}$ by Lemma \ref{lemma:y_j} and $|| U U^\top - U_0 U_0^\top || \lesssim \frac{1}{n} + \frac{1}{p_{\min}} $   by Proposition \ref{prop:recovery_P_0}.

First, we focus on the case $m=m'$. In particular, by the central limit theorem, we have 
\begin{equation*}
    \begin{aligned}
     \sqrt{n} \big(\frac{1}{n}\lambda_{0mj}^\top A_m^\top F_0^\top F_0 A_{m} \lambda_{0mj'} - \lambda_{0mj}^\top \lambda_{0mj'}\big) & =    \sqrt{n}  \big( \lambda_{0mj}^\top \frac{1}{n}F_{0m}^\top F_{0m} \lambda_{0mj'} - \lambda_{0mj}^\top \lambda_{0mj'}\big)  \\
     & = \sqrt{n} \big[\frac{1}{n}\sum_{i=1}^n (\eta_{0mi}^\top\lambda_{0mj})(\eta_{0mi}^\top \lambda_{0mj'})-\lambda_{0mj}^\top \lambda_{0mj'}\big]\\
     & \Rightarrow N(0, \Omega_{0mjj'}^2),
    \end{aligned}
\end{equation*}
where $$
\Omega_{0mjj'}^2 = \begin{cases}
    (\lambda_{0mj}^\top \lambda_{0mj'})^2 + ||\lambda_{0mj}||^2 ||\lambda_{0mj'}||^2, \quad &\text{if } j \neq j',\\
  2||\lambda_{0mj}||^4, \quad &\text{otherwise},
  \end{cases}
$$
since the $(\eta_{0mi}^\top\lambda_{0mj})(\eta_{0mi}^\top \lambda_{0mj'})$'s are independent with $E[(\eta_{0mi}^\top\lambda_{0mj})(\eta_{0mi}^\top \lambda_{0mj'})] = \lambda_{0mj}^\top \lambda_{0mj'}$ and $\text{var}[(\eta_{0mi}^\top\lambda_{0mj})(\eta_{0mi}^\top \lambda_{0mj'})] = \Omega_{0 mjj'}^2$.
For $j \neq j'$, we have 
\begin{equation*}
    \begin{aligned}
    \frac{\sqrt{n}}{n} \big(\lambda_{0mj}^\top A_m^\top F_0^\top e_{m}^{(j')}  & + \lambda_{0mj'}^\top A_{m}^\top F_0^\top e_{m}^{(j)}\big)\\
    & =  \frac{\sqrt{n}}{n}   \big( \lambda_{0mj}^\top  F_{0m}^\top e_{m}^{(j')}  + \lambda_{0mj'}^\top  F_{0m'}^\top e_{m}^{(j)}\big) \sim N(0, s_{mjj'}^2(F_0)) \overset{d}{=} s_{mjj'}(F_0) Z_{mjj'},
    \end{aligned}
\end{equation*}
where $Z_{mjj'} \sim N(0,1)$, $Z_{mjj'}$ is independent of $F_0$, and $$
s_{mjj'}^2(F_0) = \sigma_{0mj'}^2\lambda_{0mj}^\top \frac{1}{n}F_{0m}^\top F_{0m}\lambda_{0mj} + \sigma_{0mj}^2\lambda_{0mj'}^\top \frac{1}{n}F_{0m}^\top F_{0m}\lambda_{0mj'}.
$$
Let $s_{0mjj'}^2 =  \sigma_{0mj'}^2||\lambda_{0mj}||^2   + \sigma_{0mj}^2||\lambda_{0mj'}||^2$ and note that 
\begin{equation*}
    \begin{aligned}
    \frac{\sqrt{n}}{n} \big(  \lambda_{0mj}^\top A_m^\top F_0^\top F_0 A_{m} \lambda_{0mj'} &+ \lambda_{0mj}^\top A_m^\top F_0^\top e_{m}^{(j')}   + \lambda_{0mj'}^\top A_{m}^\top F_0^\top e_{m}^{(j)}\big)\\
    & =  \sqrt{n}   \big[\lambda_{0mj}^\top \frac{1}{n} F_{0m}^\top F_{0m} \lambda_{0mj'}  + s_{0mjj'} Z_{mjj'} + (s_{mjj'}(F_0)-s_{0mjj'}) Z_{mjj'}\big]\\
    & \Rightarrow N(\lambda_{0mj}^\top \lambda_{0mj'}, \Omega_{0mjj'}^2 + s_{0mjj'}^2)
    \end{aligned}
\end{equation*}
as $(s_{mjj'}(F_0)-s_{0mjj'}) Z_{mjj'} \lesssim \frac{1}{\sqrt{n}}$ with probability at least $1-o(1)$, as 
$$
(s_{mjj'}(F_0)-s_{0mjj'}) Z_{mjj'} = \frac{s_{mjj'}^2(F_0)-s_{0mjj'}^2}{s_{mjj'}(F_0)+s_{0mjj'}} Z_{mjj'}
$$
and $|s_{mjj'}^2(F_0)-s_{0mjj'}^2| \lesssim \frac{1}{\sqrt{n}}$ since $||F_{0m}^\top F_{0m} - n I_{k_m}|| /n \lesssim \frac{1}{\sqrt{n}}$ with probability at least $1-o(1)$ by Corollary 5.35 of \citet{vershynin_12}, $|s_{mjj'}(F_0)+s_{0mjj'}| \asymp 1$ as $||F_{0m}|| \lesssim \sqrt{n}$ with probability at least $1-o(1)$ by Corollary 5.35, and $|s_{mjj'}(F_0)+s_{0mjj'}| \geq s_{0mjj'} \geq c_\Lambda c_\sigma$. \\
For $j=j'$, we have 
\begin{equation*}
    \begin{aligned}
   2 \frac{\sqrt{n}}{n} \lambda_{0mj}^\top A_m^\top F_0^\top e_{m}^{(j')}  & = 2 \frac{\sqrt{n}}{n}  \lambda_{0mj}^\top  F_{0m}^\top e_{m}^{(j)} \big| F_0 \sim N(0, s_{mjj'}^2(F_0)) \overset{d}{=} s_{mjj'}(F_0) Z_{mj} \\
    \end{aligned}
\end{equation*}
where $Z_{mj} \sim N(0,1)$ and $Z_{mj}$ is independent of $F_0$, and $$
\chi_{mj}^2(F_0) = 4 \sigma_{0mj}^2\lambda_{0mj}^\top \frac{1}{n}F_{0m}^\top F_{0m}\lambda_{0mj}.
$$
Repeating the same steps as above with  $s_{0mjj}^2 = 4 \sigma_{0mj'}^2||\lambda_{0mj}||^2$ we obtain the desired result.
Next, for $m \neq m'$,  note that 
$$
\lambda_{0mj}^\top A_m^\top F_0^\top F_0 A_{m'} \lambda_{0m'j'}  = \lambda_{0mj}^\top A_m^\top B_{mm'} F_{0mm'}^{c \top}  F_{0mm'}^c B_{mm'}^\top A_{m'} \lambda_{0m'j'}
$$
for the matrix $B_{mm'} \in \mathbb R^{k_0 \times k_{mm'}^c}$ described in the proof for Theorem \ref{thm:posterior_contraction}.  
Hence, similarly as above, by the central limit theorem, we have 
\begin{equation*}
    \begin{aligned}
     \sqrt{n} \big(\frac{1}{n}\lambda_{0mj}^\top A_m^\top F_0^\top F_0 A_{m'} \lambda_{0m'j'} - \lambda_{0mj}^\top \lambda_{0mj'}\big) 
    & = \sqrt{n} \big[\frac{1}{n}\sum_{i=1}^n (\eta_{0i}^{\top}  A_m\lambda_{0mj})(\eta_{0i}^{\top} A_{m'}\lambda_{0m'j'})-\lambda_{0mj}^\top \lambda_{0m'j'}\big]\\
     & \Rightarrow N(0, \Omega_{0mm'jj'}^2),
    \end{aligned}
\end{equation*}
where $$
\Omega_{0mm'jj'}^2 =    (\lambda_{0mj}^\top A_m^\top A_{m'} \lambda_{0m'j'})^2 + ||\lambda_{0mj}||^2 || \lambda_{0m'j'}||^2.
$$ 
Analogously, we have 
\begin{equation*}
    \begin{aligned}
    \frac{\sqrt{n}}{n} \big(\lambda_{0mj}^\top A_m^\top F_0^\top e_{m'}^{(j')}  & + \lambda_{0m'j'}^\top A_{m'}^\top F_0^\top e_{m}^{(j)}\big)\\
    & =  \frac{\sqrt{n}}{n}   \big( \lambda_{0mj}^\top  F_{0m}^\top e_{m'}^{(j')}  + \lambda_{0m'j'}^\top  F_{0m'}^\top e_{m}^{(j)}\big) \sim N(0, \chi_{mm'jj'}^2(F_0)) \overset{d}{=} s_{mjj'}(F_0) Z_{mm'jj'},
    \end{aligned}
\end{equation*}
where $Z_{mmjj'} \sim N(0,1)$, $Z_{mmjj'}$ is independent of $F_0$, and $$
s_{mm'jj'}^2(F_0) = \sigma_{0m'j'}^2\lambda_{0mj}^\top \frac{1}{n}F_{0m}^\top F_{0m}\lambda_{0mj} + \sigma_{0mj}^2\lambda_{0m'j'}^\top \frac{1}{n}F_{0m'}^\top F_{0m'}\lambda_{0mj'}.
$$
Applying the same steps as for the case $m=m'$, we let 
$$s_{0mmjj'}^2 =  \sigma_{0m'j'}^2||\lambda_{0mj}||^2   + \sigma_{0mj}^2||\lambda_{0m'j'}||^2,$$ and obtain 
\begin{equation*}
    \begin{aligned}
    \frac{\sqrt{n}}{n} \big(  \lambda_{0mj}^\top A_m^\top F_0^\top F_0 A_{m'} \lambda_{0m'j'} &+ \lambda_{0mj}^\top A_m^\top F_0^\top e_{m'}^{(j')}   + \lambda_{0m'j'}^\top A_{m'}^\top F_0^\top e_{m}^{(j)}\big) \\
    & \quad \Rightarrow N(\lambda_{0mj}^\top A_m^\top A_{m'} \lambda_{0mj'}, \Omega_{0mm'jj'}^2 + s_{0mm'jj'}^2).
    \end{aligned}
\end{equation*}
\end{proof}

\begin{proof}[of Theorem \ref{thm:bvm}]
    We follows similar steps as in Theorem 3.7 of \citet{fable}.     Recall that a sample for $\tilde \lambda_{ml}$ from $\tilde \Pi$ can be represented as $\tilde \lambda_{ml} = \hat \lambda_{ml} + \frac{\rho_m \tilde \sigma_{mj}}{\sqrt{n + \tau_m^{-2}}} \tilde e_{mj}$, where $\tilde e_{mj} \sim N_{k_0} (0, I_{k_0})$, $ \tilde \sigma_{mj} \sim IG(\nu_n/2 \nu_n \delta_{mj}^2/2)$.  
    Thus, $$\tilde \lambda_{mj}^\top \tilde \lambda_{m'j'} =  \hat \lambda_{ml}^\top  \hat \lambda_{m'j'} +  \frac{\tilde \sigma_{m'j'}\rho_{m'}}{\sqrt{n + \tau_{m'}^{-2}}} \hat \lambda_{ml}^\top \tilde e_{m'j'} + \frac{\tilde \sigma_{mj} \rho_m }{\sqrt{n + \tau_m^{-2}}} \hat \lambda_{m'l'}^\top \tilde e_{mj} +  \frac{\tilde \sigma_{mj} \tilde \sigma_{m'j'} \rho_m \rho_{m'}}{\sqrt{(n + \tau_{m}^{-2})(n + \tau_{m'}^{-2})}} \tilde e_{mj}^\top \tilde e_{m'j'}. $$
    First of all, with $\tilde \Pi$ probability at least $1-o(1)$, 
    \begin{equation*}
        \begin{aligned}
            \sqrt{n}\frac{\tilde \sigma_{mj} \tilde \sigma_{m'j'} \rho_m \rho_{m'}}{\sqrt{(n + \tau_{m}^{-2})(n + \tau_{m'}^{-2})}} |\tilde e_{mj}^\top \tilde e_{m'j'}| \lesssim \sqrt{n} \frac{1}{n} \asymp 1,
        \end{aligned}
    \end{equation*}
since $\frac{1}{\sqrt{(n + \tau_{m}^{-2})(n + \tau_{m'}^{-2})}} \asymp \frac{1}{n}$, $\tilde \sigma_{mj} \lesssim 1$ and $\tilde \sigma_{m'j'} \lesssim 1 $ as a consequence of Lemma \ref{lemma:convergence_tilde_sigma}, $ |\tilde e_{mj}^\top \tilde e_{m'j'}| \leq ||\tilde e_{mj} || || \tilde e_{m'j'} || \lesssim 1$ with $\tilde \Pi$ probability at least $1-o(1)$. 
Let $\frac{1}{\sqrt{n + \tau_{l}^{-2}}} = \frac{1}{\sqrt{n}} + h_l$, where $h_l \asymp \frac{1}{n}$. 
Then,
\begin{equation*}
    \begin{aligned}
       \sqrt{n} \big( \frac{\tilde \sigma_{m'j'}\rho_{m'}}{\sqrt{n + \tau_{m'}^{-2}}} \hat \lambda_{ml}^\top \tilde e_{m'j'} + \frac{\tilde \sigma_{mj} \rho_m }{\sqrt{n + \tau_m^{-2}}} \hat \lambda_{m'l'}^\top \tilde e_{mj} \big) =&  \tilde \rho_{m'} \sigma_{m'j'} \hat \lambda_{ml}^\top \tilde e_{m'j'} + \rho_m  \tilde \sigma_{mj} \hat \lambda_{m'l'}^\top \tilde e_{mj} \\
       & + \sqrt{n}h_{m'}\tilde \rho_{m'} \sigma_{m'j'} \hat \lambda_{ml}^\top \tilde e_{m'j'} + \sqrt{n}h_m \rho_m  \tilde \sigma_{mj} \hat \lambda_{m'l'}^\top \tilde e_{mj},
    \end{aligned}
\end{equation*}
with 
\begin{equation*}
    \begin{aligned}
       |\sqrt{n}h_{m'}\tilde \rho_{m'} \sigma_{m'j'} \hat \lambda_{ml}^\top \tilde e_{m'j'} + \sqrt{n}h_m \rho_m  \tilde \sigma_{mj} \hat \lambda_{m'l'}^\top \tilde e_{mj} | \lesssim \frac{1}{\sqrt{n}} 
    \end{aligned}
\end{equation*}
with $\tilde \Pi$ probability at least $1-o(1)$.
Next, for $m, m' = 1, \dots, M$, $j=1, \dots, p_m$ and $j' = 1, \dots, p_{m'}$
$$
 T^2_{mm'jj'}(\rho_m, \rho_{m'}) = \begin{cases}
  4 \rho_m^2 \tilde \sigma_{mj}^2 ||\hat \lambda_{mj}||^2, \quad &\text{if } j=j' \text{ and } m=m',\\
    \rho_{m'}^2 \tilde \sigma_{m'j'}^2 ||\hat \lambda_{mj}||^2 + \rho_m ^2 \tilde \sigma_{mj}^2 ||\hat \lambda_{m'j'}||^2 , \quad &  \text{otherwise},
 \end{cases}
$$
and 
$$
T^2_{0mm'jj'}(\rho_m, \rho_{m'})  = \begin{cases}
4\rho_m^2 \sigma_{0mj}^2 || \lambda_{0mj}||^2, \quad &\text{if } j=j' \text{ and } m=m',\\
  \rho_{m'}^2  \sigma_{0m'j'}^2 || \lambda_{0mj}||^2 + \rho_m^2 \sigma_{0mj}^2 || \lambda_{0m'j'}||^2 , \quad &  \text{otherwise},
\end{cases}
$$
then, we have 
\begin{equation*}
    \begin{aligned}
      \tilde \rho_{m'} \sigma_{m'j'} \hat \lambda_{ml}^\top \tilde e_{m'j'} + \rho_m  \tilde \sigma_{mj} \hat \lambda_{m'l'}^\top \tilde e_{mj}  =& T_{0mjj'}(\rho_m, , \rho_{m'}) Z_{mm'jj'} \\
      &+ (T_{mm'jj'}(\rho_m, , \rho_{m'}) -  T_{0mm'jj'}(\rho_m, , \rho_{m'}))Z_{mm'jj'}
    \end{aligned}
\end{equation*}
where $Z_{mm'jj'} \sim N(0,1)$. Moreover, since $||\hat \lambda_{mj}|| \overset{pr}{\to} ||\lambda_{0mj}||$, for all $m=1, \dots, M$ and $j=1, \dots, p_m$ by Lemma \ref{lemma:norm_lambda_hat} and $\tilde \sigma_{mj}^2$ contracts around $\sigma_{0mj}^2$ under $\tilde \Pi$, for all $m=1, \dots, M$ and $j=1, \dots, p_m$, by Lemma \ref{lemma:convergence_tilde_sigma}, we have $T_{mm'jj'}(\rho_m, \rho_{m'}) -  T_{0mm'jj'}(\rho_m, , \rho_{m'}) \overset{pr}{\to} 0$. An application of Lemma F.2 of \citet{fable} concludes the proof.    
\end{proof}

\subsection{Additional Results}\label{sec:additional_results}
\begin{proposition}\label{prop:consistency_S}
      If Assumptions \ref{assumption:model}--\ref{assumption:hyperparameters} hold, and $n/\sqrt{p_{\min}} = o(1)$, 
       where $p_{\min} = \min_{m=1, \dots M}p_m$. For $m=1, \dots M$ and $j,j' = 1, \dots, p_m$, let $\hat S_{mjj'}$ be
       \begin{equation*}\label{eq:S_sq_hat}
              \hat S_{mm'jj'}^2 =    \begin{cases}
                 4\delta_{mj'}^2||\hat \lambda_{mj}||^2 + 2||\hat \lambda_{mj'}||^4 , \quad & \text{if } j=j' \text{ and } m=m',\\
                 \begin{aligned}
                     &\delta_{mj'}^2||\hat \lambda_{mj}||^2   + \delta_{mj}^2||\hat \lambda_{mj'}||^2 + (\hat \lambda_{mj}^\top \hat\lambda_{mj'})^2 \\
                     &+ ||\hat \lambda_{mj}||^2 ||\hat \lambda_{mj'}||^2, \\
                 \end{aligned}
              \quad & \text{if } j \neq j'  \text{ and } m=m',\\
              \begin{aligned}
             & \delta_{m'j'}^2||\hat\lambda_{mj}||^2   + \delta_{mj}^2||\hat \lambda_{m'j'}||^2 + (\hat \lambda_{mj}^\top  \hat \lambda_{m'j'})^2 \\
             &+  || \hat \lambda_{mj}||^2 || \hat \lambda_{m'j'}||^2,     
             \end{aligned}
 \quad & \text{otherwise}.
             \end{cases}
       \end{equation*}
       Then, as $n, p_1, \dots, p_M \to \infty$, for all $m, m'=1, \dots, M$,
       $$
        \hat S_{mm'jj'} \overset{pr}{\to}  S_{0mm'jj'},
       $$
where $ S_{0mm'jj'}$ is defined in Theorem \ref{thm:clt}.
\end{proposition}
\begin{proof}[of Proposition \ref{prop:consistency_S}]
    By Theorem \ref{thm:clt}, we have $\hat \lambda_{mj}^\top \hat \lambda_{m'j'} \overset{pr}{\to} \lambda_{0mj}^\top A_m^\top A_{m'} \lambda_{0m'j'}$ for all $j= 1, \dots, p_m $, $j' = 1, \dots, p_{m'}$ and $m, m' = 1, \dots, M$, and $\hat \lambda_{mj}^\top \hat \lambda_{mj'} \overset{pr}{\to} \lambda_{0mj}^\top  \lambda_{0mj'}$, for all $j,j'= 1, \dots, p_m $, and $m= 1, \dots, M$.  By Lemma \ref{lemma:delta_j}, we have $\delta_{mj}^2 \overset{pr}{\to} \sigma_{0mj}^2$ for all $j= 1, \dots, p_m $, and $m= 1, \dots, M$. An application of continuous mapping theorem proves the result.
\end{proof}

\subsection{Additional lemmas}\label{sec:additional_lemmas}

\begin{lemma}\label{lemma:sigma_hat}
    Let $\hat \sigma_{mj}^2  =  ||(I - U_m U_m)^\top y_m^{(j)}||^2/n$ for $j=1, \dots p_m$, $m=1, \dots, M$. Then, if Assumptions \ref{assumption:model}--\ref{assumption:hyperparameters} hold, as $n, p_m \to \infty$,  
    $$
    \hat \sigma_{mj}^2 \overset{d}{=} \frac{\sigma_{0mj}^2}{n}\chi_{n-k_m}^2 + r_{mj},
$$
where $\max_{m=1, \dots, M; j=1, \dots, p_m}|r_{mj}| \lesssim \frac{1}{n} + \frac{1}{p_{\min}}$ with probability at least $1-o(1)$,  and 
    $$\hat \sigma_{mj}^2 \overset{pr}{\to} \sigma_{0mj}^2.$$
\end{lemma}
\begin{proof}[of Lemma \ref{lemma:sigma_hat}]
    Note that \begin{equation*}
        \begin{aligned}
            (I - U_m U_m^\top) y_m^{(j)} &= (I  - U_{0m}U_{0m}^\top) y_m^{(j)} + (U_{0m}U_{0m}^\top-U_m U_m^\top) y_m^{(j)} \\
            &= (I  - U_{0m}U_{0m}^\top) \epsilon_m^{(j)} + (U_{0m}U_{0m}^\top-U_m U_m^\top) y_m^{(j)} 
        \end{aligned}
    \end{equation*}
    Thus,
    \begin{equation*}
    \begin{aligned}
         n\hat \sigma_{mj}^2 = &||(I  - U_{0m}U_{0m}^\top)\epsilon_m^{(j)}||^2 + ||(U_{0m}U_{0m}^\top-U_m U_m^\top) y_m^{(j)} ||^2 \\
        &+ 2 \epsilon_m^{(j)\top} (I  - U_{0m}U_{0m}^\top) (U_{0m}U_{0m}^\top-U_m U_m^\top) y_m^{(j)}. 
    \end{aligned}
    \end{equation*}
    In particular, $$||(I  - U_{0m}U_{0m}^\top)\epsilon_m^{(j)}||^2  \overset{d}{=} \sigma_{0mj}^2\chi_{n-k_m}^2$$
    and, with probability at least $1-o(1)$,
     \begin{equation*}
    \begin{aligned}
         ||(U_{0m}U_{0m}^\top-U_m U_m^\top) y_m^{(j)} || &\leq  ||U_{0m}U_{0m}^\top-U_m U_m^\top|||| y_m^{(j)} ||\lesssim \frac{1}{\sqrt{n}} + \frac{\sqrt{n}}{p_m},\\
        |\epsilon_m^{(j)\top} (I  - U_{0m}U_{0m}^\top) (U_{0m}U_{0m}^\top-U_m U_m^\top) y_m^{(j)}|  &\lesssim ||\epsilon_m^{(j)}|| ||I  - U_{0m}U_{0m}^\top||  ||(U_{0m}U_{0m}^\top-U_m U_m^\top) y_m^{(j)} || \lesssim 1 + \frac{n}{p_m}.
    \end{aligned}
    \end{equation*}
since $  ||U_{0m}U_{0m}^\top-U_m U_m^\top|| \lesssim \frac{1}{n} + \frac{1}{p_m}$, $||y_m^{(j)}|| \asymp ||\epsilon_m^{(j)}|| \asymp \sqrt{n}$, with probability at least $1-o(1)$, by Proposition \ref{prop:reovery_P_m} and Lemma \ref{lemma:y_j} respectively.
Combining all of the above, we get 
$$
\hat \sigma_{mj}^2 \overset{d}{=} \frac{\sigma_{0mj}^2}{n}\chi_{n-k_m}^2 + r_{mj},
$$
where $\max_{m=1, \dots, M; j=1, \dots, p_m}|r_{mj}| \lesssim \frac{1}{n} + \frac{1}{p_m}$ with  probability at least $1-o(1)$. The second result follows from a straightforward application of the weak law of large numbers. 
\end{proof}

\begin{lemma}\label{lemma:delta_j}
    Let $\delta_{mj}^2$ be the scalar defined in \eqref{eq:posterior_params} for $j=1, \dots p_m$, $m=1, \dots, M$. Then, if Assumptions \ref{assumption:model}--\ref{assumption:hyperparameters} hold, as $n,p_1, \dots, p_M \to \infty$,  
    $$
   \delta_{mj}^2 \overset{d}{=} \frac{\sigma_{0mj}^2}{n}\chi_{n-k_0}^2 + s_{mj},
$$
where $\max_{m=1, \dots, M; j=1, \dots, p_m}|s_{mj}| \lesssim \frac{1}{n} + \frac{1}{p_{\min}}$ with probability at least $1-o(1)$,  and 
    $$\delta_{mj}^2\overset{pr}{\to} \sigma_{0mj}^2.$$
\end{lemma}

\begin{proof}[of Lemma \ref{lemma:delta_j}]
    Note that \begin{equation*}
        \begin{aligned}
        \delta_{mj}^2 =&  \frac{1}{n} || (I - U_0 U_0^\top) y_m^{(j)}||^2 + \frac{\nu_0\delta_0^2}{\nu_n} + \frac{1}{n}y_m^{(j)\top}(U_0 U_0^\top - U U^\top)  y_m^{(j)} + \frac{1}{n}  \big(1 - \frac{n}{n+\tau_m^{-2}}\big) y_m^{(j)\top} U U^\top y_m^{(j)} \\
        &- \big(\frac{1}{\nu_n} + \frac{1}{n}\big) y_m^{(j)\top} \big(I_n - U U^\top\big) y_m^{(j)} 
        \end{aligned}
    \end{equation*}
    where 

    \begin{equation*}
    \begin{aligned}
      \frac{1}{n} || (I - U_0 U_0^\top) y_m^{(j)}||^2 =  \frac{1}{n} || (I - U_0 U_0^\top) \epsilon_m^{(j)}||^2  &\sim \frac{\sigma_{0mj}^2}{n} \psi_m, \quad\psi_m  \sim \chi_{n-k_0}^2,\\
        \frac{\nu_0\delta_0^2}{\nu_n} &\lesssim \frac{1}{n},\\
\max_{j=1, \dots, p_m}       \frac{1}{n} \big|y_m^{(j)\top}(U_0 U_0^\top - U U^\top)  y_m^{(j)} \big| & \lesssim  \big(\frac{1}{n} + \frac{1}{p_{\min}}\big),\\
\max_{j=1, \dots, p_m}      \frac{1}{n}  \big(1 - \frac{n}{n+\tau_m^{-2}}\big) \big|y_m^{(j)\top} U U^\top y_m^{(j)}\big|  &\lesssim \frac{1}{n^2} \sqrt{n} = \frac{1}{n^{3/2}},\\
   \max_{j=1, \dots, p_m}   \big(\frac{1}{\nu_n} + \frac{1}{n}\big)\big| y_m^{(j)\top} \big(I_n - U U^\top\big) y_m^{(j)} \big|& \lesssim \frac{1}{n^2} n = \frac{1}{n},
    \end{aligned}
    \end{equation*}
   since, with probability at least $1-o(1)$, $  ||U_{0}U_{0}^\top-U U^\top|| \lesssim 1/n + 1/p_{\min}$ by Proposition \ref{prop:recovery_P_0}, $\max_{j=1, \dots, p_m} ||y_m^{(j)}||  \asymp \sqrt{n}$ by Lemma \ref{lemma:y_j} respectively, and $1 - \frac{n}{n+\tau_m^{-2}} \asymp \frac{1}{n}$.
The second result follows from a straightforward application of the weak law of large numbers. 
\end{proof}

\begin{lemma}\label{lemma:y_j}
    Suppose Assumptions \ref{assumption:model}--\ref{assumption:sigma} hold. Then, with probability at least $1-o(1)$,
    \begin{equation*}
       \max_{j=1, \dots, p_m} ||e_m^{(j)}|| \lesssim \sqrt{n}, \quad \max_{j=1, \dots, p_m} ||y_m^{(j)}|| \lesssim \sqrt{n}, \quad (m=1, \dots, M).
    \end{equation*}
\end{lemma}
\begin{proof}[of Lemma \ref{lemma:y_j}]
The first result follows from $\max_{j=1, \dots, p_m} || {e}_m^{(j)} || \lesssim \sqrt{n}$ Lemma 1 of \citet{laurent_massart} combined with Assumption \ref{assumption:sigma}. 
Next, consider the following 
\begin{equation*}
    \begin{aligned}
        ||y_m^{(j)}|| \leq ||F_{0m}|| ||{\lambda}_{0mj}|| + ||{e}_m^{(j)}||.
    \end{aligned}
\end{equation*}
Recall that, with probability at least $1-o(1)$,
$||{F}_{0m}|| \asymp \sqrt{n}$ by Corollary 5.35 of \citet{vershynin_12} and $\max_{j=1, \dots, p_m}  ||\mathbf{\lambda}_{0mj}|| \leq ||{{\Lambda}}_0 ||_\infty \sqrt{k_m} \asymp 1$ by Assumption \ref{assumption:Lambda}. The result follows.
\end{proof}

\begin{lemma}\label{lemma:Y_m}
    Suppose Assumptions \ref{assumption:model}--\ref{assumption:sigma} hold. Then, with probability at least $1-o(1)$,
    \begin{equation*}
     ||E_m|| \lesssim \sqrt{n} + \sqrt{p_m}, \quad   ||Y_m|| \lesssim \sqrt{np_m}, \quad (m=1, \dots, M).
    \end{equation*}
\end{lemma}
\begin{proof}[ of Lemma \ref{lemma:Y_m}] 
Let $\sigma_{m, max}^2 = \max_{j=1, \dots, p_m} \sigma_{0mj}^2$ and $\sigma_{m, sum}^2 = \sum_{j=1}^{p_m} \sigma_{0mj}^2 \leq p_m \sigma_{m, max}^2$. Then, by Corollary 3.11 in \citet{bandeira_van_handel_16}, we have $||E_m|| \lesssim \sqrt{n} \sigma_{m, max} + \sigma_{m, sum} \leq  \sigma_{m, max}(\sqrt{n} + \sqrt{p_m})$, with probability at least $1-o(1)$. The first result follows from Assumption \ref{assumption:sigma}. 
For the second result, note that $||F_m \Lambda_{0m}^\top|| \lesssim \sqrt{n p_m}$ with probability at least $1-o(1)$, since $||\Lambda_{0m}|| \asymp \sqrt{p_m}$ by Assumption \ref{assumption:Lambda} and $||F_m|| \asymp \sqrt{n}$ with probability at least $1-o(1)$ by Corollary 5.35 in \citet{vershynin_12}.
\end{proof}

\begin{lemma}\label{lemma:Lambda_m_hat}
Let $\hat \Lambda_m$ be the estimator defined in \eqref{eq:mu_Lambda}. Then, with probability at least $1 - o(1)$, 
\begin{equation*}
    ||\hat \Lambda_m|| \lesssim \sqrt{p_m}, \quad (m=1, \dots, M).
\end{equation*}
\end{lemma}
\begin{proof}[of Lemma \ref{lemma:Lambda_m_hat}]
Consider
    \begin{equation*}
         ||\hat \Lambda_m|| \leq \frac{\sqrt{n}}{n + \tau_m^{-2}} ||Y_m|| ||U|| \lesssim \frac{1}{\sqrt{n}} \sqrt{n p_m} = \sqrt{p_m},
    \end{equation*}
    since $ ||Y_m|| \lesssim \sqrt{n p_m}$ with probability at least $1-o(1)$ by Lemma \ref{lemma:Y_m}.
\end{proof}

\begin{lemma}\label{lemma:norm_lambda_hat}
 Let $\hat \lambda_{mj}$ be the posterior mean for $\lambda_{mj}$ defined in \eqref{eq:posterior_params}. Then, if Assumptions \ref{assumption:model}--\ref{assumption:hyperparameters} hold and $n/\sqrt{p_{\min}} = o(1)$, 
       where $p_{\min} = \min_{m=1, \dots M}p_m$, for all $j,j'=1, \dots, p_m$ and $m =1, \dots, M$, we have 
 \begin{equation*}
     \begin{aligned}
         ||\hat \lambda_{mj}|| &\overset{pr}{\to} ||\lambda_{0mj}||,\\
         \hat \lambda_{mj}^\top \hat \lambda_{mj'}  &\overset{pr}{\to}\lambda_{0mj}^\top  \lambda_{0mj'}.  
     \end{aligned}
 \end{equation*}
Moreover, for all $j=1, \dots, p_m$, $j'=1, \dots, M$ and $m,m' =1, \dots, M$, with $m \neq m'$, we have 
 \begin{equation*}
     \begin{aligned}
         \hat \lambda_{mj}^\top \hat \lambda_{m'j'}  \overset{pr}{\to}\lambda_{0mj}^\top A_m^\top A_{m'} \lambda_{0m'j'}.  
     \end{aligned}
 \end{equation*}
\end{lemma}
\begin{proof}[of Lemma \ref{lemma:norm_lambda_hat}]
    The results follows from Theorem \ref{thm:clt}.
\end{proof}

\begin{lemma}\label{lemma:convergence_tilde_sigma}
     Suppose Assumptions \ref{assumption:model}--\ref{assumption:sigma} hold and $\sqrt{n}/p_m = o(1)$ for all $m=1,\dots, M$. Then, as $n, p_1, \dots, p_M \to \infty$, we have \begin{equation*}
       \max_{j=1, \dots, p_m} |   \tilde \sigma_{mj}^2 - \sigma_{0mj}^2| \lesssim \big(\frac{\log{p_m}}{n}\big)^{1/3} + \frac{1}{p_{\min}}, \quad (m=1, \dots, M),
\end{equation*}
with probability at least $1-o(1)$,  where $\tilde \sigma_{mj}^2$ is a sample for the idiosyncratic error variance of the $j-$th variable in the $m$-th view from $\tilde \Pi$ and $p_{\min} = \min(p_1, \dots, p_M)$.
\end{lemma}
\begin{proof}[of Lemma \ref{lemma:convergence_tilde_sigma}]
We follow similar steps of the proof of part (b) of Theorem 3.6 in \citet{fable}. Consider the following decomposition
\begin{equation*}
    \begin{aligned}
       \tilde \sigma_{mj}^2 - \sigma_{0mj}^2 &= \big(1 - 2 \frac{U_{mj}}{\nu_n} \big) \big\{\delta_{mj}^2- \sigma_{0mj}^2 - 2\frac{U_j}{\nu_n} \sigma_{0j}^2 \big\},
    \end{aligned}
\end{equation*}
where $U_{mj}= \frac{\nu_n}{2} \big(\delta_{mj}^2 \tilde \sigma_{mj}^{-2} -1 \big)$. With probability at least $1-o(1)$, we have $\max_{j=1,\dots, p_m} |U_{mj}|/\mathbf{\gamma}_n \lesssim (\log{p_m}/n)^{1/3}$, by Lemma E.7 in \citet{fable}, which implies $\min_{j=1,\dots, p_m}\big|1 + \frac{2}{\mathbf{\gamma}_n} U_{mj}\big| \gtrsim 1/2$ with probability at least $1-o(1)$.
Thus,
\begin{equation*}
    \max_{j=1, \dots, p_m} |\tilde \sigma_{mj}^2 - \sigma_{0mj}^2| \lesssim  \max_{j=1, \dots, p_m} | \delta_{mj}^2 - \sigma_{0mj}^2| + \sigma_{0mj}^2 \max_{j=1, \dots, p_m} \frac{|U_{mj}|}{\mathbf{\gamma}_n}.
\end{equation*}
Consider the following decomposition
\begin{equation*}
    \begin{aligned}
       \delta_{mj}^2- \sigma_{0mj}^2 &=  \frac{\sigma_{0mj}^2}{n} \big( \zeta -  (n-k) \big) + \frac{k}{n}\sigma_{0mj}^2  + s_{mj}, \quad \zeta \sim \chi_{n-k_m}^2
    \end{aligned}
\end{equation*}
where we used the representation of $\delta_{mj}$ of Lemma \ref{lemma:delta_j} and, with probability at least $1-o(1)$, $\max_{j=1, \dots, p_m} |s_{mj}| \lesssim \frac{1}{n} + \frac{1}{p_{\min}}$ by Lemma \ref{lemma:delta_j}. 
Moreover, with probability at least $1-o(1)$, 
 \begin{equation*}
         \max_{j=1, \dots, p_m} \bigg| \frac{C_j}{\sigma_{0mj}^2} - (n-k_m)  \bigg|
 \lesssim n \big(\frac{\log{p_m}}{n}\big)^{1/3},   \end{equation*}
 by Lemma E.7 of \citet{fable}.
   Combining all of the above, we get
\begin{equation*}
       \max_{j=1, \dots, p_m} |   \tilde \sigma_{mj}^2 - \sigma_{0mj}^2| \lesssim \big(\frac{\log{p_m}}{n}\big)^{1/3} + \frac{1}{p_{\min}},
\end{equation*}
with probability at least $1-o(1)$.
\end{proof}

\section{Additional details about the selection of the number of latent factors}\label{sec:additional_details_k}
To implement the criterion in \eqref{eq:k_hat}, we approximate the joint maximum likelihood estimates as follows. For each value of the latent dimension $k$, we estimate the latent factors active in the $m$-th view by the matrix of left singular vectors of $Y_m$ scaled by $\sqrt{n}$, and the factor loadings through the penalized regression   
\begin{equation*}
     \hat \Lambda_m^{(k)}  = \underset{ B_m \in \mathbb R^{p_m \times k}}{\operatorname{argmin}} || Y_m - \hat F_m^{(k)} B_m^\top ||_F^2 + \tau_m^{-2} ||B_m||_F^2,
\end{equation*}
where $\hat F_m^{(k)}$ is the estimate of the latent factors and $\tau_m$ is chosen using the criterion described in \ref{subsec:hyperparams}. Finally, we estimate the residual error variances through the empirical variances of columns $Y_m - \hat F_m^{(k)}  \hat \Lambda_m^{(k) \top}$.

\section{Additional details about the numerical experiments}\label{sec:additional_details_simulations}
We fit \texttt{ROTATE} setting $\lambda_0 = 5$, $\lambda_1 = 0.001$, and $\epsilon=0.05$, as in the example provided in the supplemental code of \citet{rotate}.
For both \texttt{MOFA} and \texttt{ROTATE}, we set the upper-bound to the number of factors to $k_0 +5$. 
For \texttt{AJIVE}, we set each view's signal rank to the true value $k_m$, thereby giving it an advantage over competitors. To choose the rank of the shared signal, let $UDV^\top $ be the singular value decomposition of $U_{\text{joint}} = \begin{bmatrix}
    U_1 & \cdots & U_M
\end{bmatrix} \in \mathbb R^{n \times \sum_m k_m}$. We set the rank of the shared signal to $$
\underset{j=1,\dots, \min\{k_1, \dots, k_M\}}{\operatorname{argmin}} \frac{d_{j}}{d_{j+1}},
$$
where $d_l$ is  the $l$-th entry on the diagonal of $D$.
As for the alternative spectral estimator (\texttt{SVD}), we estimate the low-rank component of the intra-view covariance via $V_m D_m^2 V_m^\top /n$ and the inter-view covariances via $V_m D_mU_m^\top U_{m'} D_{m'} V_{m'}^\top / n$. For all methods that do not provide an estimate of the residual error variances, we estimate it via the empirical variances of the unexplained variation.

For the coverage results shown in Table \ref{tab:uq}, we report the average coverage of the entries off the diagonal for the intra-view covariance since \texttt{FABLE} does not directly provide samples for $\Lambda_m \Lambda_m$ but only for $\Lambda_m \Lambda_m + \Psi_m$.

The code to implement the \texttt{FAMA} methodology and replicate the experiments is available at \url{https://github.com/maurilorenzo/FAMA}. All experiments were run on a Laptop with 11th Gen Intel(R) Core(TM) i7-1165G7 @ 2.80GHz and 16GB RAM.
\section{Additional details about the application}\label{sec:additional_details_application}
For \texttt{MOFA} and \texttt{ROTATE}, we consider two options for the upper bound of the number of factors: $\hat k_0^{(1)} +5$ and $\hat k_0^{(2)} +5$, where $\hat k_0^{(1)}$ and $\hat k_0^{(2)}$ are the estimates of the general latent factors used by \texttt{FAMA} obtained via the criterion in \eqref{eq:k_0_hat} and by \texttt{FABLE} respectively. 
\texttt{MOFA} does not directly provide an estimate for the residual error variances. Therefore, we estimated the $j$-th element on the diagonal of $\Psi_m$ via the empirical variance of $y_m^{(j)} - \hat F_m \hat \lambda_{mj}$, where $\hat F_m$ and $\lambda_{mj}$ are the \texttt{MOFA} estimates for the factors active in view $m$ and the corresponding loadings for the $j$-th outcome.

Applying the criterion in \eqref{eq:k_hat}, for each view $m$, we set $k_{m, \max}$ at the smallest value such that the first $k_{m, \max}$ principal components explain at least $90\%$ of the variation of the data in this view. 

We pre-processed the data by transform each variable as $\tilde y_{mij} = \Phi^{-1}(\hat F_{mj}(y_{ij}^{m}))$, where $\hat F_{mj}(\cdot)$ is the empirical cumulative distribution function of the $j$-th variable in the $m$-th view.

\section{Bayesian decoupling for inference on specific vs shared latent factors}\label{sec:bd}
An important task arising in the analysis of multi-view data is inferring which factors are active in each view, and, similarly, for each pair of views, which factors are shared and which are view-specific. Our estimation approach detailed in Section \ref{sec:method} bypasses this task to avoid some of the computational and statistical hurdles associated to identifying where latent factors are active. To amend this, we develop an approach based on Bayesian decoupling (\texttt{BD}). \texttt{BD} was originally proposed in \citet{bd}, and subsequently developed in, for instance, \citet{fold, bd_li}. In summary, \texttt{BD} is a decision theoretic framework that starts by inferring the posterior distribution of the parameters, often without sparsity or other constraints that are desired for interpretable inference but not need to accurately characterize the data likelihood, and then obtains estimates by minimizing the posterior expectation of a loss. The loss trades-off fidelity to the Bayesian posterior with desirable properties of the estimate, such as sparsity. To infer the factors that are active in the $m$-th view, $F_m \in \mathbb R^{n \times k_m}$, we let 
\begin{equation}\label{eq:bd_F_m}
     (\bar F_m, \bar \Lambda_m )  \in \underset{\bar \Lambda_m \in\in \mathbb R^{p_m \times k_m}, F_m \in \mathbb R^{n \times k_m}~ :~ F_m^\top F_m = I_n}{\operatorname{argmin}} E\big[|| F_m \Lambda_m^\top - \hat F\tilde \Lambda_m^\top ||_F^2 \mid Y],
\end{equation}
where $k_m$ is the number of latent factors for view $m$, $\hat F$ is the estimate of the overall latent factors obtained with the procedure described in Section \ref{subsec:factor_estimation} and the expectation is taken with respect to the posterior distribution of $\tilde \Lambda_m$, which is defined in \eqref{eq:lambda_sigma_cc}. \eqref{eq:bd_F_m} is  intuitive in that it minimizes the \textit{a posteriori} reconstruction error of the estimated signal while enforcing the correct rank for the $m$-th view. The solution to \eqref{eq:bd_F_m} is not unique, since for a solution $(\bar F_m, \bar \Lambda_m )$, any tuple of the type $(\bar F_m R, \bar \Lambda_m R)$ for some orthogonal marix $R$ is a solution as well. Such rotational ambiguity is typical of factor models. The problem in \eqref{eq:bd_F_m} is equivalent to 
\begin{equation*}\label{eq:bd_F_m_2}
     (\bar F_m, \bar \Lambda_m )  \in \underset{\bar \Lambda_m \in\in \mathbb R^{p_m \times k_m}, F_m \in \mathbb R^{n \times k_m}~ :~ F_m^\top F_m = I_n}{\operatorname{argmin}} || F_m \Lambda_m^\top - \hat F\hat \Lambda_m^\top ||_F^2,
\end{equation*}
where $\hat \Lambda_m $ is the posterior expectation of $\tilde \Lambda_m$, which is given by \eqref{eq:mu_Lambda}. Letting $D_m$ be the diagonal matrix with the leading $k_m$ singular values of $\hat F \hat \Lambda_m^\top$ on the diagonal and $\hat U_m$, $\hat V_m$ be the matrices of corresponding left and right singular vectors, a solution to \eqref{eq:bd_F_m} is given by $\bar F_m = \sqrt{n} \hat U_m$ and $\bar \Lambda_m = \hat V_m \hat D_m / \sqrt{n}$, and $\bar F_m$ provides an estimate of $F_m$. $ (\bar F_m, \bar \Lambda_m )$ corresponds to the PCA-estimate of latent factors and loadings using $\hat F\hat \Lambda_m^\top$ as the data matrix. Hence, it benefits from the signal denoising step (estimating of the low-rank signal), which pools information across views via the steps in Section \ref{subsec:factor_estimation}. This can improve accuracy in particular for smaller or noisier views, where estimates using only view-specific data tend to have low accuracy. If inferring shared vs view-specific latent factors is of interest, once estimates of the matrices $\{\bar F_m\}$ are obtained via the step above, one could apply any off-the-shelf method, such as \texttt{AJIVE}, to identify shared and specific component for any collection of views.  This inference problem becomes massively easier by using estimates for $F_m's$ compared to working with the original data matrices. 

We compare the estimate proposed in \eqref{eq:bd_F_m} with the estimates of the latent factors obtained via \texttt{SVD} (which is equivalent to the PCA-estimator \citep{bai_03} applied to each view separately) and \texttt{AJIVE}. We use the same data generating mechanism for the numerical experiments in Section \ref{sec:numerical_experiments}. Since our focus is solely on comparing the accuracy of different methods in estimating the latent factors, we set the number of active latent factors in each view equal to its true value for every method to avoid confounding due to rank misspecification. We measure estimation accuracy via the relative Frobenius Procrustes error averaged across views (reported in Table \ref{tab:factors}). Our proposed methodology achieves a better performance than the alternatives. 

\begin{table}[ht]
\caption{Average relative Frobenius Procrustes error in estimating the latent factors. Values are computed over 50 independent replications. We report mean and standard deviation (sd) for each method and sample size. Values have been multiplied by $10^2$. \label{tab:factors}}
\centering
\footnotesize
\resizebox{0.7\textwidth}{!}{%
\begin{tabular}{lcccccc}

& \multicolumn{6}{c}{Balanced scenario} \\
$n$ 
& \multicolumn{2}{c}{\texttt{FAMA}} 
& \multicolumn{2}{c}{\texttt{AJIVE}} 
& \multicolumn{2}{c}{\texttt{SVD}} \\
& mean & sd & mean & sd & mean & sd \\
250  & 16.18 & 0.43 & 27.73 & 11.99 & 19.48 & 0.40 \\
500  & 12.96 & 0.37 & 35.25 & 14.20 & 16.22 & 0.31 \\
1000 & 10.56 & 0.18 & 45.65 & 0.37  & 14.35 & 0.16 \\
& \multicolumn{6}{c}{Unbalanced scenario} \\
$n$ 
& \multicolumn{2}{c}{\texttt{FAMA}} 
& \multicolumn{2}{c}{\texttt{AJIVE}} 
& \multicolumn{2}{c}{\texttt{SVD}} \\
& mean & sd & mean & sd & mean & sd \\
250  & 19.07 & 0.39 & 30.09 & 10.54 & 24.51 & 0.40 \\
500  & 16.01 & 0.32 & 38.19 & 12.54 & 21.28 & 0.31 \\
1000 & 13.90 & 0.15 & 45.63 & 0.37  & 14.36 & 0.16 \\

\end{tabular}%
}
\end{table}


\end{document}